\documentclass[12pt]{article}

\usepackage[utf8]{inputenc} 
\usepackage[T1]{fontenc}

\usepackage{newtxtext}
\usepackage{amsmath,amssymb,amsthm,cite,enumitem,mathtools,xcolor}
\usepackage[nosymbolsc,cmintegrals]{newtxmath}
\usepackage{bm} 
\usepackage[normalem]{ulem}  

\usepackage[margin=3cm]{geometry}
\newcommand\blfootnote[1]{%
  \begingroup
  \renewcommand\thefootnote{}\footnote{#1}%
  \addtocounter{footnote}{-1}%
  \endgroup}

\title{Quantum general covariance}

\author{Christian Gaß$^{1}$, José M. Gracia-Bondía$^{2,3}$,
and Karl-Henning Rehren$^{4}\quad$%
\blfootnote{Email: christian.gass@fuw.edu.pl, jmgb@unizar.es,
krehren@uni-goettingen.de}
\blfootnote{ORCID: CG: \href{https://orcid.org/0000-0003-4059-6817}{https://orcid.org/0000-0003-4059-6817}, JGB: \href{https://orcid.org/0000-0002-8036-4589}{https://orcid.org/0000-0002-8036-4589},
  KHR:
  \href{https://orcid.org/0000-0003-3640-6515}{https://orcid.org/0000-0003-3640-6515}}
\\[6pt] {\footnotesize $^1$Department of Mathematical Methods in
Physics, Faculty of Physics,} 
\\
{\footnotesize University of Warsaw, Pasteura 5, 02-093 Warszawa,
Poland.} 
\\
{\footnotesize $^2$ Laboratorio de Física Teórica y Computacional,
Universidad de Costa Rica, San Pedro 11501, Costa Rica.}
\\
{\footnotesize $^3$ CAPA and Departamento de Física Teórica,
Universidad de Zaragoza, 50009 Zaragoza, Spain.} 
\\
{\footnotesize $^4$Institut für Theoretische Physik, 
Georg-August-Universität Göttingen, 37077 Göttingen, Germany.}
}

\date{\today}

\DeclareMathOperator{\supp}{supp}   

\newcommand{\al}{\alpha}            
\newcommand{\be}{\beta}             
\newcommand{\dl}{\delta}            
\newcommand{\eps}{\varepsilon}      
\newcommand{\Ga}{\Gamma}            
\newcommand{\ka}{\kappa}            
\newcommand{\la}{\lambda}           
\newcommand{\sg}{\sigma}            
\newcommand{\Th}{\Theta}            

\newcommand{\lra}{\leftrightarrow}  
\newcommand{\lrd}{\overset{\leftrightarrow}{\partial}}
\newcommand{\ld}{\overset{\leftarrow}{\partial}}
\newcommand{\pa}{\partial}

\newcommand{\bN}{\mathbb{N}}        
\newcommand{\bR}{{\mathbb{R}}}      
\newcommand{\bS}{\mathbb{S}}        

\newcommand{\sO}{\mathcal{O}}       

\newcommand{\go}{\mathfrak{g}}      

\newcommand{\del}{\partial}         
\newcommand{\dc}{\delta_c}          
\newcommand{\doublcont}[1]{(\!( #1 )\!)} 
\newcommand{\modd}{\stackrel{\textup{div}}{=}} 
\newcommand{\wt}{\widetilde}        
\def\bb#1{\doublcont{#1}}

\newcommand{\aside}[1]{} 

\newcommand{\vev}[1]{\langle\!\langle#1\rangle\!\rangle} 
\newcommand{\Tren}{T_{\textup{ren}}} 
\newcommand{\word}[1]{\quad\text{#1}\quad} 

\def\wick:#1:{\,\mathopen:#1\mathclose:\,} 
\def\ol{\overline}
\def\epsi^#1_#2{\eps^{#1}{}_{\!#2}} 
\def\epsii_#1^#2{\eps_{#1}{}^{\!#2}} 
\def\lLa^#1_#2{\Lambda^{#1}{}_{#2}}  

\def\duo<#1,#2>{\langle#1,#2\rangle} 
\def\scal<#1|#2>{\langle#1\mathbin|#2\rangle} 

\theoremstyle{plain}
\newtheorem{theorem}{Theorem}
\numberwithin{theorem}{section}

\newtheorem{corollary}[theorem]{Corollary}

\newtheorem{definition}[theorem]{Definition}

\newtheorem{lemma}[theorem]{Lemma}
\newtheorem{notation}[theorem]{Notation}

\newtheorem{proposition}[theorem]{Proposition}
\newtheorem{remark}[theorem]{Remark}

\newcommand{\sref}[1]{Sect.~\ref{#1}}
\newcommand{\aref}[1]{App.~\ref{#1}}

\def\sfrac#1#2{\hbox{\large{$\frac{#1}{#2}$}}}

\makeatletter
\renewcommand{\section}{\@startsection{section}{1}{\z@}%
                       {-3.5ex \@plus -1ex \@minus -.2ex}%
                       {2.3ex \@plus.2ex}%
                       {\normalfont\large\bfseries}}
\renewcommand{\subsection}{\@startsection{subsection}{2}{\z@}%
                       {-3.25ex \@plus -1ex \@minus -.2ex}%
                       {1.5ex \@plus .2ex}%
                       {\normalfont\normalsize\bfseries}}
\makeatother

\numberwithin{equation}{section}

\usepackage{hyperref}
\usepackage{xcolor}
\hypersetup{
    colorlinks,
    linkcolor={blue!40!black},
    citecolor={blue!40!black},
    urlcolor={blue!40!black}
}

\begin{document}

\maketitle

\medskip

\begin{abstract}
The structure of quantum interactions with fields of helicity two
(``gravitons'') is strongly constrained by three principles:
positivity (Hilbert space), covariance, and locality of observables. To fulfil them
simultaneously, some (non-observable) fields need to be non-local. We work with string-localized fields. The
results then follow from the condition that entities closely related
to observables, like the $\bS$-matrix, be local and
string-independent. They in particular reproduce the interactions
dictated by general covariance in classical field theory.
Graviton-matter couplings are consistent only when the graviton
self-interaction is taken into account as well.
\end{abstract}

\medskip

\begin{flushright}
\textit{Science is not a collection of truths; it is a continuous 
exploration of mysteries}

\medskip

-- Freeman Dyson
\end{flushright}

\medskip


\medskip

\baselineskip10pt
\tableofcontents
\baselineskip14.7pt

\medskip

\section{Introduction}
\label{sec:introibo}

\subsection{The background}
\label{sec:backofthemind}

Physics being above all an experimental science, the development of
quantum gravity (QG), as compared to other branches of theoretical
physics, has been hampered by a relative poverty of phenomenological
development \cite{PPNPpaper} and the widespread belief that QG
becomes dominant only at Planck scales. On the other hand, there seems
to be in the literature a large, if not total, consensus that the
(Einstein equations and the) Hilbert and/or the Einstein’s actions for
general relativity can be re-derived by means of quantum field theory.
The idea of coupling a massless helicity $|h|=2$ particle
(``graviton'') to its own stress-energy-momentum tensor, and the
program of iterating the process, go back to Kraichnan and Gupta in
the 1950s. They were elaborated by Feynman \cite{Feynman62, FMW} and
Deser \cite{Deser:1969, Boulware:1974}, with coworkers. A more
sophisticated later version is \cite{vanderBij:1981} -- leading to
Einstein's unimodular action. Similar techniques were put to work in
\cite{Nemausa}, in spite of prior pointed and detailed criticism of
the whole idea by Padmanabhan~\cite{Padma:2004}.

To the same ends, the book \cite{ScharfLast} by Günter Scharf employs
in the framework of ``causal gauge invariance'' the inductive
development by Epstein and Glaser of Bogoliubov’s recursive
$\bS$-matrix theory – that is, the functional analogue of the
Dyson series \cite{BogoliubovS80, EpsteinGlaser73,DuetschP}. It was based on
previous partial work by Scharf and Wellmann~\cite{ScharfWell}, and on
more advanced one of the same kind by Schorn \cite{Schorn}. In
\cite{ScharfLast} the first two non-trivial terms in the Newton's
constant expansion of the Einstein--Hilbert action for general
relativity are recovered; standard lore~\cite{Wyss} has it that this
is sufficient to determine all higher-order terms, and so to the
recover the whole theory.

Now, in the respect of the still elusive search for QG, there remains
a chasm, since whereas \textit{classical} linear or non-linear systems
can always be set on a Hilbert space by use of Koopman operators
\cite{Koopman31, GayBalmazT22}, that is structurally impossible in the
setting of gauge theories. The purpose of the present article is to
revisit the problem from the viewpoint and with the more rigorous
tools of string-localized quantum field theory (SQFT), which manages
as well to incorporate the Bogoliubov–Epstein–Glaser method under the
guise of the \textit{string-independence} (SI) condition. This turns
out to be a fruitful strategy, leading to the concept of ``quantum
general covariance'', exploited in this paper, dispelling the
perceived inconsistency between general relativity and quantum field
theory. The main aim of SQFT is to avoid the difficulties associated
with ``canonical quantization'' of gauge theories, notably the use of
state spaces with indefinite metrics, by working with interaction
densities which are operators on a physical Hilbert space throughout
-- so that observables are elements of an operator
algebra with a \textit{positive} vacuum functional. This is possible with the help of \textit{string-localized}
fields -- here above all the quantum counterpart $h_{\mu\nu}(x;c)$ of
the metric deviation field or ``graviton'' -- arising as integrals
over the associated field strengths -- here $F_{[\mu\ka][\nu\la]}(x)$,
essentially the linearized Riemann--Christoffel curvature tensor --
see Eq.~\eqref{eq:RF}. For recent accounts of SQFT, consult
\cite{Phocaea, MundRS23}.

The dependence of the interaction density on the ``string'' is a total
derivative, so that the action (an integral over the density) is
string-independent, and so string-localization does not spoil locality
of a perturbative QFT model, at least in first order. The latter
property of the action appears to be a rather restrictive one;
however, it is satisfied by all interactions among particles of the
Standard Model. Here we intend to show that it can be satisfied also
both for the self-interaction of gravitons and for the interactions
between gravitons and matter.

On the face of it, the total-derivative property is sufficient to
ensure SI of the action only at first order of perturbation theory. At
higher orders, there will occur ``obstructions'' to~SI. These must be
well-localized, so that they have a chance to be canceled by
higher-order interactions -- then we say ``the obstruction is
resolved''. Both the conditions on the interaction at the lower order
ensuring that resolution, and the form of the higher-order
interactions themselves, are therefore predictions of the approach.
They do not refer to an underlying classical theory, nor do they
assume gauge or diffeomorphism invariance, but arise as consistency
conditions for local and covariant quantum field models on
Hilbert space.%
\footnote{The need for unobservable fields assigned to infinite
string-like spacelike regions for the {\em mathematical description}
of charged particle states was recognized by Buchholz and Fredenhagen
over forty years ago \cite{BF82}, see our \sref{sec:Conclusion}. In
spite of the apparent non-locality, the principle of locality of
observables is held, and they derive localization properties of
particle states which hold in all standard models of relativistic
quantum theory.}

\subsection{Quantum general covariance}
\label{sec:QGC}

Remarkably, for all the interactions of interest, one finds oneself in
the ``lucky'' situation of well-localized obstructions. In other
words, quantum field theory in the SQFT dispensation has the power to
predict on its own the structure of consistent interactions. In the
presence of interacting massive vector bosons, like in the Standard
Model (SM) that structure includes chirality of the weak interaction
\cite{Rosalind}, the Lie-algebra makeup of cubic self-couplings of
several vector bosons and the quartic self-coupling familiar from
Yang--Mills theory \cite{Borisov}, the need of a scalar field with its
potential terms \cite{MundRS23} -- without ever invoking ``gauge
invariance'' nor classical ``Lagrangians''.

\smallskip

The above is shown here to hold also true for gravity: the unique
self-interactions obtained from the expansion of the classical
diffeomorphism invariant action of general relativity at first and
second order in terms of the Newton constant are coincident with the
self-interactions of our $h^{\mu\nu}(x;c)$. So we reproduce the
uniqueness of graviton self-interactions which is the general
conclusion of the various approaches presented in the beginning. The
main new feature in SQFT is the nature of the argument, totally based
on quantum principles.

\smallskip

We find a second, even more remarkable result, concerning the
couplings to matter. Free matter fields in a flat spacetime have
conserved stress-energy tensors \cite{MundRS17a, MundRS17b}. In
perturbation theory, their time-ordered products fail to be conserved,
i.e., Ward identities are violated. Now, the coupling of matter to a
field of helicity 2 turns out to have an obstruction proportional to
the violation of the Ward identity. This obstruction cannot be
resolved on its own. However, on including the self-interaction of the
graviton field, one finds an interference term yielding another
obstruction. And both obstructions together can be resolved by a
higher-order matter coupling. That the latter coincide with the
second-order term of the expansion of the generally covariant matter
coupling is in the vein of the above: consistency of the quantum
couplings ``predicts'' general covariance.

This outcome (at second order of perturbation theory) reveals a kind
of ``predestined match'' between matter and gravity. On the one hand,
the violation of the matter's Ward identity knows nothing about
gravity. On the other hand, the self-interaction of gravitons knows
nothing about matter. Yet, the two taken together produce a sum of
non-resolvable obstructions that together can be resolved -- like a
lock and a key. Consistency is only achieved when the key and the
lock, apparently ignorant of each other, do meet. We are tempted to
call this mechanism, which does not require classical covariance:
\textbf{``Quantum general covariance''}. The latter holds for matter
consisting of an arbitrary number of scalar, photon and Dirac fields.
We prove these facts here in due course.

\subsection{Strategy of the paper}
\label{sec:qui-tacet}

In order to arrive at our results, the overall strategy schematically
is: (i) Start with a first order input interaction and compute its
obstruction to~SI. (ii) Make sure that the obstruction is resolvable.
This is separately the case for the self-interaction of gravitons, but
not for the coupling to matter. It is again true if both are taken
together. (iii) Determine the respective second-order interactions that
do the job. If necessary, continue to the next order.

\smallskip

We next elaborate a bit on how SQFT enters this strategy in the
subject at hand -- leaving the filling in of the technical details and
proofs for the main body of the paper.

\begin{enumerate}\itemsep0mm

\item 
Consider the field-strength tensor $F_{[\mu\ka][\nu\la]}(x)$ of the
massless particle of helicity 2. It is defined on the Fock space over
the direct sum of unitary representations of the Poincar\'e group
corresponding to helicities $|h|=2$. This Hilbert space therefore
contains only physical graviton states. Its two-point function is
given by
\begin{align}
\label{eq:tpF2a}
\vev{F_{[\mu\ka][\nu\la]}F'_{[\rho\tau][\sg\pi]}} =
\sfrac12&\big[\eta_{\mu\rho}\eta_{\nu\sg}+
\eta_{\mu\sg}\eta_{\nu\rho}-\eta_{\mu\nu}\eta_{\rho\sg}\big]
\pa_\ka\pa_\la\pa'_\tau\pa'_\pi W_0(x-x')
\notag
\\
&~\begin{array}{c} 
- (\mu\leftrightarrow\ka)\,\,
\\
[-1mm]-(\nu\leftrightarrow\la)\,\,
\end{array}
\quad
\begin{array}{c} 
-(\rho\leftrightarrow\tau)\,\,
\\
[-1mm]-(\sg\leftrightarrow\pi);
\end{array}
\end{align}
where $W_0$ denotes the two-point function of a massless scalar field
and $\eta_{\mu\nu}$ is the (mostly negative) Minkowski metric.

\item 
Represent $F_{[\mu\ka][\nu\la]}$ as:
\begin{align*}
F_{[\mu\ka][\nu\la]}(x) = [\pa_\mu\pa_\nu h_{\ka\la}(x;c) -
\pa_\ka\pa_\nu h_{\mu\la}(x;c) - \pa_\mu\pa_\la h_{\ka\nu}(x;c) +
\pa_\ka\pa_\la h_{\mu\nu}(x;c)].
\end{align*}
Here $h_{\mu\nu}(x;c)$ is a traceless symmetric field tensor defined
on the physical Fock space by ``string integrations'' $I_c$ over $F$
-- see Eq.~\eqref{eq:def_SL_pot}. This means that our graviton field
$h_{\mu\nu}(x;c)$ is localized on a spacelike cone emanating from $x$
and extending to infinity, with a directional profile given by a
function $c(e)$ of the spacelike directions~$e$.%
\footnote{As an idealization, one may think of a narrow cone as a
``(half-)string'', i.e., a line extending from $x$ to infinity. A
gauge-theoretic motivation for the introduction of the strings will be
given in \aref{a:BRST}. But that motivation is not ours'. We need the
field $h_{\mu\nu}(x;c)$ in order to formulate interactions on
\textit{Hilbert space}. The key point is that the physics ought not
depend on $c$.}

\item 
The prerequisite for a \textit{string-independent} $\bS$-matrix for
gravity is a string-independent first-order action. That is, we
require an interaction density $L_1(x;c)$ whose string variation
$\dl_cL_1(x;c)$ is a total derivative:
\begin{align}
\label{LQ}
\delta_c \big(L_1\big)  = \pa_\mu Q_1^\mu.
\end{align}
This structure is called an {\em ``$L$-$Q$ pair''} -- see
\sref{sec:LQ} below. In the cases at hand, $L_1$ is a Wick polynomial
in $h(x;c)$ and its derivatives, and the~$Q_1^\mu$ are Wick
polynomials involving as well an auxiliary field~$w$ measuring the
variation of~$h$ with respect to the string. The smallest possible UV
scaling dimension of $L_1$ is {\em five} (reflecting the
non-renormalizability of self-interactions with helicity~2 by power
counting); with this dimension the form of~$L_1$ is essentially
unique, being cubic in $h$ with two derivatives.

\item
Property \eqref{LQ} ensures that the perturbative $\bS$-matrix
$\bS=\mathbf{1}+i\kappa \int d^4x\, L_1(x;c)+\dots$, where
$\ka=4\sqrt{2\pi G}$ for $G$ the Newton constant, be
string-independent in first order. 

But SI in general is not assured to hold in second order, because
$TL_1(x_1;c)L_1(x_2;c)$ in general fails to be a total derivative:
\begin{align}
\label{S11}
\delta_c\bigl(TL_1(x_1;c)L_1(x_2;c)\bigr) =&~T\pa_{x_1^\mu}
Q_1^\mu(x_1;c)L_1(x_2;c)+(x_1\leftrightarrow x_2)
\\
\label{S12} 
\neq&~\pa_{x_1^\mu}T Q_1^\mu(x_1;c)L_1(x_2;c) +(x_1\leftrightarrow x_2).
\end{align}
The \textit{difference} $O^{(2)}(x_1,x_2,c)$ between $\eqref{S11}$ and
$\eqref{S12}$ is the \textit{second-order obstruction} to SI of the
$\bS$-matrix.

\item
The computation of the obstruction at the tree level turns out to be
straightforward in terms of the propagators. The result must be of the
form:
\begin{align}
\label{eq:SI2}
O^{(2)}(x_1,x_2;c) = i\delta_c\big(L_2(x_1;c)\big)\delta(x_1-x_2) +
\hbox{derivatives}
\end{align}
with $L_2$ being a quartic Wick polynomial.

\item
If \eqref{eq:SI2} holds, add $L_2$ to the interaction density:   
\begin{align}
\label{L}
L_{\rm int}(x;c) = \kappa L_1(x;c) + \frac{\kappa^2}2 L_2(x;c) + \cdots
\end{align} 
This ensures that in second order, the contribution from $L_2$ to the
$\bS$-matrix cancels (resolves) the obstruction $O^{(2)}$ up to
derivatives. In other words, the $\bS$-matrix including the {\em
``induced'' interaction} $L_2$ is string-independent in second order.

\item
Proceed recursively: compute the obstruction in third order of the
$\bS$-matrix with interaction $\kappa L_1+\frac12\kappa^2L_2$. It must be
of the form:
\begin{align}
\label{SI3}
O^{(3)}(x_1,x_2,x_3;c)= i^2\delta_c\bigl(L_3(x_1;c)\bigr)
\delta(x_1-x_2)\delta(x_2-x_3) + \hbox{derivatives}.
\end{align}
This determines $L_3$ (up to a derivative), and so on. In practice,
these higher-order corrections are rarely required.
\end{enumerate}

Note that \eqref{eq:SI2} is a rather strong requirement, without which
the construction would be worthless. The subtraction defining
$O^{(2)}$ is convenient because \eqref{S11} is supported everywhere
in~$M\times M$, whereas the difference is a priori supported on the
set $\{x_1\in x_2+\bR\supp c\}$. That Eq.~\eqref{eq:SI2} may be
satisfied at all with \emph{localized} support on the set
$\{x_1=x_2\}$ in all interactions of interest is the hallmark of SQFT.
In order to decide on \eqref{eq:SI2} and determine $L_2$ up to
derivatives, it is sufficient to compute \eqref{S11}, up to
derivatives. We exploit this simplicity in the subsequent analysis.

Our first main result is that the cubic graviton self-interaction
$L_1$ satisfying \eqref{LQ} is unique up to a total divergence, and
that its second-order obstruction can be resolved. Also $L_1$ and
$L_2$ as determined by \eqref{eq:SI2} reproduce the classical
expansion of the Einstein action in powers of the metric deviation
$h_{\mu\nu}$ defined by $\go^{\mu\nu}\equiv \sqrt{-g} g^{\mu\nu}=
\eta^{\mu\nu}+ \kappa h^{\mu\nu}$, where the classical field
$h^{\mu\nu}(x)$ is replaced by the string-localized quantum field
$h^{\mu\nu}(x;c)$.

The second main result is that, given the stress-energy-momentum
tensor $\Theta^{\mu\nu}$ of a matter field, the coupling $L_{1,\rm
mat} = \frac12 h_{\mu\nu}(x;c)\Theta^{\mu\nu}$ to the latter -- which is
the unique possibility from the $L$-$Q$-pair condition discussed in item~3.
above -- has an obstruction that cannot be resolved; namely
\eqref{eq:SI2} fails. But this obstruction can be resolved when one
adds the self-coupling of the graviton field: \textbf{the key opens
the lock}. Moreover, the resulting second-order induced matter
interaction reproduces the classical expansion of the generally
covariant ``matter Lagrangian''.

\subsection{Final introductory comments}
\label{sec:keep-talking}

An extra word is in order concerning renormalizability. The field
strength $F$ has scaling dimension 3, the string-localized field $h$
-- as any other boson potential in SQFT -- has scaling dimension 1. In
spite of the latter's good behaviour, the first-order interaction has
dimension~5, as we shall see. Thus it is not UV power-counting
renormalizable.%
\footnote{That power counting remains a meaningful notion in SQFT was
shown in \cite{G22}.}
Our work in this paper does not need to address loop diagrams.
Nevertheless, one faces an infinite series of higher-order
interactions of increasing dimension already at tree level. This is
not a surprise, given that also classical general covariance implies
an infinite power series. Whether counterterms for UV-divergent loop
diagrams can be absorbed in a renormalization of the unique coupling
constant $\kappa$, or substantially change the structure of the
higher-order interactions, cannot be assessed as of yet.

Our path to perturbative QG in a Minkowski background does not branch
from standard gauge theory in the first place; rather it has to be
regarded as an alternative to it. In the former, the BRST method is a
crucial tool to eliminate unphysical degrees of freedom after the
perturbation development has been done. Such tools we do not need,
since from the beginning we work on a physical Hilbert space with the
correct degrees of freedom. Nevertheless, it is instructive to discuss
the relation between the two approaches. In \aref{a:BRST} we show how
SQFT modifies the BRST approach in such a way that the interaction
density whose BRST variation is a divergence becomes invariant -- and
so BRST invariance becomes obsolete altogether. Quantum general
covariance manifested by the ``lock-key scenario'' discussed above, is
in fact a characteristic feature of the Hilbert space approach;
whereas the non-resolvable obstructions of the matter interaction
(Sect.\ \ref{ssc:h_SET}) is removed in the BRST approach with
the help of ghost fields, blurring the power of quantum principles.

\section{Preliminaries on string-localized fields}
\label{sec:Preliminaries}

Standard lore has that the redundancy of degrees of freedom in gauge
theories can be avoided if one allows the potentials of certain
quantum fields to be non-local -- see for example the discussion in
the textbook~\cite[Ch.~22]{Schwartz14}. The mildest and better
physically justified way of introducing such a non-locality is by
using the \emph{string-localized} potentials originally mooted
in~\cite{MundSY04, MundSY06}. String-localized potentials can be
defined for arbitrary masses $m\geq0$ and spins/helicities $s\ge1$.
Here we mainly deal with the massless potential field corresponding to
helicity~2, baptized the \emph{string-localized graviton potential}.
In the following we list its relevant properties, introducing all the
pertinent concepts and notations.

\subsection{The string-localized graviton potential}
\label{sec:slocpot}

Consider the massless free \textit{field strength} tensor
$F_{[\mu\ka][\nu\la]}(x)$ of helicity 2 -- essentially, as advertised,
the linearized curvature tensor. It is defined in terms of
the creation and annihilation operators for the physical helicities
$h=+2$ and $h=-2$ on the Fock space over the corresponding unitary
Wigner representation of the Poincar\'e group with $m=0$ and $|h|=2$.
By construction it satisfies the wave equation and enjoys the
symmetries and field equations:
\begin{align}
&F_{[\mu\ka][\nu\la]} = -F_{[\ka\mu][\nu\la]},\;\;
F_{[\mu\ka][\nu\la]} = F_{[\nu\la][\mu\ka]},\;\; \eta^{\mu\nu}
F_{[\mu\ka][\nu\la]} = 0,\;\; \del^\mu F_{[\mu\ka][\nu\la]}=0,
\label{eq:properties_F}
\\
&F_{[\mu\ka][\nu\la]} + F_{[\ka\nu][\mu\la]} + F_{[\nu\mu][\ka\la]}
=0, \quad \del_\rho F_{[\mu\ka][\nu\la]} + \del_\mu
F_{[\ka\rho][\nu\la]} + \del_\ka F_{[\rho\mu][\nu\la]} =0,
\notag
\end{align}
its positive-definite two-point function being given already in
\eqref{eq:tpF2a}. In order to introduce interactions one needs a
potential~$h_{\mu\nu}$, from which the field
tensor~$F_{[\mu\ka][\nu\la]}$ arises as the double exterior 
derivative or linearized curvature. Now, the construction of a local
pointlike potential as an operator-valued distribution on the same
Fock space is impossible~\cite{Weinberg64b}. However, there is a {\em
  string-localized} potential satisfying (2.5), which does live on the
same Fock space \cite{MundSY04,MundSY06}. For its definition \ref{def:SL_pot}, we need the
the operation of string integration.

\begin{definition}
\label{def:Ic}
Let $H:=\{ e\in\bR^{1+3}\;|\; e^2<0 \}$ denote the open subset of
spacelike vectors in Minkowski space and let $c\in C_c^\infty(H)$ be a
smooth test function with compact support so that the cone $\{\bR_+
\cdot \supp c\}$ is convex, satisfying the unit weight condition
$\int_H d^4e\,c(e) = 1$. We define the \emph{string integration}
$I_c^\mu$ by its action on a generic (possibly operator-valued)
distribution $f$ on the Minkowski space:
\begin{align}
I_c^\mu f(x) := \int_H d^4e \, c(e) e^\mu \int_0^\infty d s\, f(x+se)
\equiv \int_H d^4e \, c(e) \int_{C_{x,e}} f(y) \,dy^\mu,
\end{align}
where the curve $C_{x,e}$ is the ray (``string'') in direction $e$
extending from $x$ to spatial infinity.
\end{definition} 

It is obvious that
\begin{align}
\label{eq:Ic_inverse_del}
I_c^\mu \del_\mu f(x) = \del_\mu I_c^\mu f(x) = -f(x), \word{or}
I_c^\mu \del_\mu =: (I_c\del) = -1,
\end{align}
whenever derivates and integrals may be interchanged and there are no
boundary terms at infinity.

\begin{definition}
\label{def:SL_pot}
The \emph{string-localized graviton potential} $h_{\mu\nu}(x;c)$ is
defined via
\begin{align}
\label{eq:def_SL_pot}
h_{\mu\nu}(x;c) := I_c^\ka I_c^\la F_{[\mu\ka][\nu\la]}(x). 
\end{align}
\end{definition}
Thus it ``lives'' on the Hilbert-Fock space of the $F$-tensor field
\cite{MundSY04,MundSY06} and possesses the same two degrees of freedom
as $F$ -- see~\cite{MundSY06,MundRS17b} for the general spin/helicity
cases.

The scaling dimension of quantum fields and their propagators dictates
the strength of UV divergences. Because the two-point function for~$F$
in \eqref{eq:tpF2a} scales under $x\to\lambda x$ like $\lambda^{-6}$,
the field $F$ has scaling dimension~3; and because of two
string-integrations in its definition~(2.4), $h_{\mu\nu}(x;c)$ has
scaling dimension~1, the same as the scaling dimension of the
canonically quantized gauge potential $h_{\mu\nu}(x)$, with the
difference that the former is defined on the Hilbert space, while the
latter is not.

The cubic interaction coupling will consequently have scaling
dimension~5, which is beyond the power-counting bound for
renormalizability, as in all other approaches. This is to be expected
since the gravitational coupling constant $\kappa$ has a negative mass
dimension.%
\footnote{Recall however the comment in Sect.~\ref{sec:keep-talking}
on prospects that ``infinitely many UV renormalization constants"
might be all fixed by the condition of string-independence of the
$\bS$-matrix.}

The $\bS$-matrix constructed in terms of the string-localized
potential will be manifestly unitary, because all involved fields
act on a Hilbert space. The issue with the $\bS$-matrix will rather be its
string-independence  -- see Sect.~\ref{sec:LQ}.

It is also clear that $h_{\mu\nu}(c)$ depends on
the choice of $c$, being localized along the string (or
rather cone) $x+\bR_+ \supp c$. However, as already pointed out, its double
exterior derivative gives back the original point-localized
$F$-tensor,
\begin{align}
\label{eq:Fcurlh}
\del_\mu \del_\nu h_{\ka\la}(x;c) - \del_\mu\del_\la h_{\ka\nu}(x;c) -
\del_\ka \del_\nu h_{\mu\la}(x;c) + \del_\ka \del_\la h_{\mu\nu}(x;c) =
F_{[\mu\ka][\nu\la]}(x),
\end{align}
which is independent of $c$. The string-localized potential satisfies
the following relations, as consequences of the properties
\eqref{eq:properties_F} of the $F$-tensor:
\begin{align}
\label{eq:properties_h}
h_{\mu\nu}=h_{\nu\mu}, \quad \eta^{\mu\nu} h_{\mu\nu} = 0,\quad
\del^\mu h_{\mu\nu}=0, \quad \square h_{\mu\nu}=0,\quad I_c^\mu
h_{\mu\nu}=0.
\end{align}

In the sequel we shall omit in the notation the dependence on $c$; it
is sufficient to know that the dependence of $h$ on $c$ is such that
when the profile $c(e)$ is varied, one has
\begin{align}
\label{eq:delta_c_h}
\delta_c \big(h_{\mu\nu}\big) = \del_\mu w_\nu+ \del_\nu w_\mu, \quad
\word{where} w_\mu=(\delta_cI_c^\ka) h_{\ka\mu}
\end{align}
is another string-localized field. The precise form of the
string-integrated differential operator $\delta_c I_c^\ka$ will not be
relevant.%
\footnote{One has to vary the function $c$ by a function of
weight zero. That is, $\delta c(e)= \del_{e^\tau} b^\tau(e)$. This
gives
\begin{align*}
\delta_c I_{c}^\ka f(x) &= \int_H e^\ka d^4e \, \del_{e^\tau}
b^\tau(e) \int_0^\infty ds\; f(x+se) =-\int_H d^4e\, b^\tau(e)
\int_0^\infty ds \,\Big(f(x+se) \delta^\ka_\tau + s \del_\tau f(x+se)
e^\ka \Big) .
\end{align*}
All that matters is that
$\delta_c I_c^\ka \del_\ka=0$ by \eqref{eq:Ic_inverse_del}, hence
\begin{align*}
 \delta_c\big(h_{\mu\nu}\big)=\delta_c(I_c^\ka I_c^\la)F_{[\mu\ka][\nu\la]} 
 = (\delta_cI_c^\ka) \pa_{[\mu} h_{\ka]\nu}+ (\mu\leftrightarrow\nu) =
(\delta_cI_c^\ka) \del_\mu h_{\ka\nu}+ (\mu\leftrightarrow\nu)
=:\del_\mu w_\nu+\del_\nu w_\mu .
\end{align*}
If $\delta c(e)$ is compactly supported, then one may choose $b(e)$ to
be compactly supported, and $w_\mu$ is string-localized in the convex
hull of $\mathbb{R}_+ \cdot \bigl(\supp(b)\cup \supp(c)\bigr)$.}
By \eqref{eq:properties_h}, the field $w_\mu$ is divergenceless:
\begin{align}
\del^\mu w_\mu=0.
\end{align}
That the string variation \eqref{eq:delta_c_h} be a derivative is
essential in order to construct $\bS$-matrices $T\exp\bigl(i\int
d^4x\,L_{\rm int}(x;c)\bigr)$ that do not depend on the profile
function~$c$, from the interactions satisfying \eqref{LQ}.

\subsection{The $L$-$Q$ pair formalism and the
condition of string-independence}
\label{sec:LQ}

In string-localized QFT, the perturbative $\bS$-matrix can be defined
as a formal power series
\begin{align}
\label{eq:sloc_Smatrix}
\bS[g;c] = 1 + \sum_{n=1}^\infty \frac{1}{n!} \int d^4x_1 \dots \int
d^4x_n\, g(x_1) \dots g(x_n)\, S_n(x_1,\dots,x_n;c),
\end{align}
where $g$ is a test function over Minkowski space and the $S_n$ are
time-ordered products of a normal-ordered interaction density,
possibly plus local terms supported on the ``small diagonal'':
$\Delta_n := (x_1=x_2=\dots=x_n)$. Such local terms supported on
$\Delta_n$ either correspond to new Wick-polynomials in the fields or
to corrections to already existent ones. We call the new
Wick-polynomials \emph{induced interactions}, while the corrections to
already existent ones correspond to a renormalization of their
prefactor (that is, they are \emph{local counterterms}).

The expressions $S_n$ depend on the string-smearing function $c$ via
the contained string-localized fields. The first order term $S_1(x;c)$
is related to the $\sO(g)$-part of the interaction density,
\begin{align}
S_1(x;c) \equiv i L_1(x;c).
\end{align}At higher orders, the situation becomes more complicated. In general, we 
will have 
\begin{align}
S_2(x_1,x_2;c) = i^2\,\Tren[L_1(x_1;c) L_1(x_2;c)] + i L_2(x_1;c)
\delta(x_1-x_2),
\end{align}
where the renormalized time-ordered product $\Tren[L_1(x_1;c)
L_1(x_2;c)]$ contains the local counterterms and $L_2(x_1;c)$ is an
induced interaction. Then the third order contribution will be a
time-ordered product
\begin{align}
S_3(x_1,x_2,x_3;c) &=i \Tren[ S_2(x_1,x_2;c) S_1(x_3;c)] +
\textup{symmetric} 
\notag \\
&\quad + i L_3(x_1;c) \delta(x_1-x_2) \delta(x_2-x_3),
\end{align}
and so on. The full (unrenormalized) interaction density of the model
would then be:
\begin{align}
L[g;c](x) := \sum_{k=1}^\infty \frac{g(x)^k}{k!} L_k(x;c),
\end{align}
with normal-ordered induced interactions $L_k$. For the interactions
of the Standard Model, $L_k=0$ for all $k>2$, and so one does not have
to look for induced couplings after third order. For graviton
interactions, it is expected that the series does not terminate.

The dependence of scattering amplitudes on the test function $c$ (or
the existence of a preferred direction, if the support of $c$ is very
narrow) is not observed in scattering experiments. Hence, one requires
that the $\bS$-matrix be independent of the choice of $c$ in the
\emph{adiabatic limit} $g(x)\uparrow\mathrm{const}$, where the test
function $g$ becomes a true coupling constant. That is,
\begin{align}
\label{eq:SI_Smatrix}
\lim_{g\uparrow\mathrm{const}} \delta_c \big(\bS[g;c]\big) \stackrel{!}{=} 0.
\end{align}
Requirement \eqref{eq:SI_Smatrix} is what we call the \emph{SI
condition} for the $\bS$-matrix. Provided that the adiabatic limit can
be performed, it is fulfilled if
\begin{align}
\label{eq:SI_T}
\delta_c \big(S_n(x_1,\dots,x_n;c)\big) = i^n \sum_{k=1}^n
\frac{\del}{\del x_k^\mu}\,Q_{n}^\mu(x_1,\dots,x_n;c), \quad \forall
n\in\bN,
\end{align}
so that after integration by parts, one finds
\begin{align*}
\delta_c \big(\bS\big) &= - \sum_{n=1}^\infty \frac{i^n}{n!}\int
d^4x_1\dots \int d^4x_n\,\sum_{k=1}^n g(x_1) \dots (\del_\mu g)(x_k)
\dots g(x_n) \,Q_{n}^\mu
\stackrel{g\uparrow\mathrm{const}}{\longrightarrow} 0.
\end{align*}

As already indicated in \eqref{LQ}, at first order in perturbation
theory the SI condition implies the $L$-$Q$-pair condition $\delta_c
\big(L_1(x;c)\big) = \del_\mu Q_1^\mu(x;c)$. That is, the string
variation of the interaction density must be a total divergence. In
this paper we only consider tree graph contributions up to second
order of perturbation theory. 
In particular, we will not have to deal with local counterterms, but
only with induced interactions. By Wick's theorem, modulo loop graphs
we have:
\begin{align}
\label{eq:T_2nd_expansion}
\Tren[L_1(x_1;c)L_1(x_2;c)] = \wick:L_1(x_1;c)L_1(x_2;c): +
\wick:\underbracket{L_1(x_1;c)L_1}(x_2;c): +\dots,\qquad 
\\
\word{where} \wick:\underbracket{L_1(x_1;c)L_1}(x_2;c):\, :=
\sum_{\varphi,\chi} \wick: \frac{\del L_1}{\del \varphi}(x_1;c) 
\vev{\Tren \varphi(x_1;c)  \chi(x_2;c) } \frac{\del L_1}{\del \chi}(x_2;c) :
\label{eq:keyformula}
\end{align}
is the sum of all terms with one contraction. The fields
$\varphi,\chi$ may depend on the string-smearing function~$c$, and
$\vev{\Tren \varphi(x_1) \chi(x_2) }$ denotes the renormalized
time-ordered two-point function of $\varphi$ and $\chi$, whose precise
meaning is explained in the next subsection \ref{sec:propagators}.
Under application of~$\delta_c$, the first term on the right-hand side
of \eqref{eq:T_2nd_expansion} is a total divergence by virtue of the
$L$-$Q$-pair condition. Then the SI condition is satisfied at second-order
tree level if
\begin{align}
\delta_c\bigl(\wick:\underbracket{L_1(x_1;c)L_1}(x_2;c):\bigr) =
\sum_{k=1}^2\frac{\del}{\del x_k^\mu}\,Q_{2,\textup{tree}}^\mu
(x_1,x_2;c) + i\delta_c \big(L_{2,\textup{tree}}(x_1)\big)\delta(x_1 -
x_2).
\end{align}
For our present aims, the precise form of $Q_{n}^\mu$,
$Q_{n,\textup{tree}}^\mu$ and other total divergences is not required.

\subsection{Two-point functions and propagators}
\label{sec:propagators}

All two-point functions and propagators (i.e.~time-ordered two-point
functions) can be respectively expressed in terms of the two-point
functions and propagators of scalar fields $\phi$ of mass $m\ge0$:
\begin{align}
\label{eq:tpprop}
\vev{\phi(x)\phi(x')}=\int \frac{d^4k} {(2\pi)^3}\,\delta(k^2-m^2)
\theta(k^0) \,e^{-ik(x-x')}=:&~W_m(x-x'), 
\\
\vev{T\phi(x)\phi(x')}=i \int\frac{d^4k}{(2\pi)^4}
\frac{e^{-ik(x-x')}}{k^2-m^2+i\eps}=:&~D_{m,F}(x-x').
\end{align}
The two-point function for the Maxwell field strength tensor (of
helicity 1) is
\begin{align}
\label{eq:tpF1}
\vev {F_{\mu\nu}(x)F_{\ka\la}(x')} = - \eta_{\mu\ka}\pa_\nu \pa'_\la
W_0(x-x') - (\mu\lra\nu)-(\ka\lra\la).
\end{align}
The two-point function for the field strength tensor of helicity 2 was
given in \eqref{eq:tpF2a}. It is of course consistent with the
relations~\eqref{eq:properties_F} and the wave equation.
Likewise, the ``kinematic'' propagator is obtained by replacing the
massless two-point function $W_0$ by the massless propagator
$D_{0,F}$.

Similarly, the kinematic propagator for $h_{\mu\nu}$ is defined by
applying the string integrations as in \eqref{def:SL_pot}. The result
is:
\begin{align}
\label{eq:T0hh}
\vev{T_0h_{\mu\nu}(x)h_{\rho\sg}(x')} =&~\sfrac12\bigl[E'_{\mu\rho}E'_{\nu\sg}
+E'_{\mu\sg}E'_{\nu\rho}-E_{\mu\nu}E''_{\rho\sg}\bigr] D_{0,F}(x-x'),
\end{align}
where (with $I_c'$ the string integrations w.r.t.\ $x'$):
\begin{align} 
E_{\mu\nu} :=~&\eta_{\mu\nu} + I_{c,\mu}\del_\nu+ I_{c,\nu}\del_\mu+
I^2_c\del_\mu \del_\nu,
\notag
\\
E''_{\rho\sg} :=~&\eta_{\rho\sg} + I'_{c,\rho}\del'_\sg+
I'_{c,\sg}\del'_\rho+ I'^2_c\del'_\rho\del'_\sg,
\notag
\\
E'_{\mu\rho} :=~&\eta_{\mu\rho} + I'_{c,\mu}\del'_\rho+
I_{c,\rho}\del_\mu+ (I_cI_c')\del_\mu \del'_\rho,
\label{eq:E}
\end{align}
cf.~\cite{MundRS17b} in this respect. Kinematic propagators of
derivatives of $h_{\mu\nu}$ are obtained by applying such
derivatives in Eq.~\eqref{eq:T0hh}.

The kinematic propagator \eqref{eq:T0hh} and its derivatives respect
the symmetry of $h_{\mu\nu}$, but because of \begin{align} \square
D_{0,F}(x-x')=-i\delta(x-x')
\end{align}
they do not respect the trace condition in \eqref{eq:properties_h} --
consult \eqref{eq:trace_viol} in this respect. This would seem to
imply that the propagators reflect a non-existent trace degree of
freedom -- an unacceptable feature in a Hilbert space theory.
Fortunately, the linear dependency among field components can be
restored by exploiting the freedom of renormalization of propagators
(see again \aref{a:Propagators}). The matter generally consists in
adding string-integrals of derivatives of $\delta(x-x')$:
\begin{align}
\vev{T_{\rm ren}\phi(x)\chi(x')} &= \vev{T_0\phi(x)\chi(x')} +
\vev{T_r\phi(x)\chi(x')}, 
\notag
\\
\word{where} \quad\vev{T_r\phi(x)\chi(x')} &=
\sum_{\underline\mu,\underline\nu,\underline\la}
iC^{\phi,\chi}_{\underline\mu,\underline\nu,\underline\la}
I_c^{\underline\mu}{I'}_c^{\underline\nu}\,\pa^{\underline\la}\,
\delta(x-x'),
\label{eq:freedom_of_reno}
\end{align}
with numerical coefficients $C$. $\underline\lambda$ etc.\ are Lorentz multi-indices. The lengths $\vert\underline\mu\vert$ and
$\vert\underline\nu\vert$ give the numbers of string-integrations, and
$4+\vert\underline\la\vert- \vert\underline\mu\vert
-\vert\underline\nu\vert$ is the scaling dimension of the
renormalization term. By the standard power counting argument, these
numbers cannot exceed those of $\vev{T_0\phi\chi'}$. Since we do not
want to break scale invariance of a massless theory already at the
tree level, we shall actually demand {\em equality} of the scaling
dimensions.

The coefficients in \eqref{eq:freedom_of_reno} can be adjusted so that
the trace conditions are restored. Further constraints on these
coefficients will arise from the condition that the second-order
obstruction has the form~\eqref{eq:SI2} without string-integrated
$\delta$-functions. For the details, see \aref{a:Propagators}.

The renormalized propagators determine the obstructions to SI.
Propagators for matter fields and their renormalizations are studied
in \sref{ssc:h_SET}.

\section{Graviton couplings on Hilbert space}
\label{sec:gravitons}

We set out now to derive up to second order of perturbation theory in
the framework of SQFT the structure of graviton couplings to the
stress-energy tensors of the \textit{scalar, Maxwell and Dirac}
``matter'' fields, as well as the structure of the graviton
self-coupling.

As already indicated, to describe spin one-half fermions and scalars
the SQFT construction is unnecessary, their stress-energy tensors thus
being point-localized from the outset. In the case of Maxwell fields,
one might envisage coupling the string-localized graviton potential to
a string-localized stress-energy tensor for helicity 1. It turns out
however~\cite{GPhD} that the $L$-$Q$-pair condition \eqref{LQ} can only
hold when the string-localized graviton is coupled to the well-known
\textit{point-localized} stress energy-tensor of the Maxwell field.

\begin{notation}
\label{def:notation} 
To simplify notation in what follows, we drop the colon-notation for
Wick-ordering and write $A\equiv A(x)$ and $A'\equiv A(x')$ for
generic fields $A$ (possibly Wick polynomials), and $\varphi$, $\chi'$
for generic linear fields. It will be also convenient to write
\begin{align*}
(AB)_{\mu\nu} := A_{\mu\ka} B^\ka_\nu, \quad \doublcont{AB} := 
A_{\ka\la} B^{\ka\la}
\end{align*}
for contractions of symmetric tensors $A$ and $B$, and similarly for
just one tensor or products of more than two tensors: i.e.,
$\doublcont{A}$ is the trace of $A$. We will write $A\modd B$ when
expressions $A$ and $B$ differ by a total divergence.
\end{notation}

\subsection{Graviton coupling to stress-energy tensors}
\label{ssc:h_SET}

For the moment, let $\Th^{\mu\nu}_\textup{mat}$ denote the (symmetric
conserved, point-localized) stress-energy tensor associated with a
generic ``matter'' field, and consider the coupling:
\begin{align}
\label{eq:coup_matter}
L_{1,\textup{mat}} := \sfrac12 \doublcont{ h \Th_\textup{mat}}.
\end{align}
In view of \eqref{eq:delta_c_h}, clearly $L_{1,\textup{mat}}$ of above
satisfies the $L$-$Q$-pair condition:
\begin{align}
\delta_c\big(L_{1,\textup{mat}}\big) = \del_\mu Q_{1,\textup{mat}}^\mu
\word{with} Q_{1,\textup{mat}}^\mu = w_\nu\,\Th^{\mu\nu}_\textup{mat}.
\end{align}
From Eq.~\eqref{eq:keyformula}, the renormalized second order tree
graph can be expanded as
\begin{align}
\label{eq:TL1mL1m:a}
\underbracket{L_{1,\textup{mat}}L}\!_{1,\textup{mat}} &=
\sfrac{1}{4} \Th^{\mu\nu}_\textup{mat} \vev{\Tren h_{\mu\nu}
h'_{\rho\sigma}} {\Th'}^{\rho\sigma}_\textup{mat} +\sfrac{1}{4}
h_{\mu\nu} h'_{\rho\sigma} \sum_{\varphi,\chi'} \frac{\partial
\Th^{\mu\nu}_\textup{mat} }{\partial \varphi} \vev{\Tren \varphi \chi'}
\frac{\partial {\Th'}^{\rho\sigma}_\textup{mat}}{\partial \chi'}.
\end{align}
We next have to keep an eye on the renormalization of the graviton
propagators in \aref{a:Propagators}. By inspection of \eqref{eq:T0hh}
and \eqref{eq:E}, we see that all string-dependent parts of the
kinematic propagator contain uncontracted derivatives. Because the
stress-energy tensors are conserved, this means that they only
contribute divergences to the first term in \eqref{eq:TL1mL1m:a}. To
ensure the same for the renormalized propagator $\vev{T_rhh'}$, as
well it must contain only terms with uncontracted derivatives, or
string-independent terms. The latter cannot occur by the bounds
discussed around Eq.~\eqref{eq:freedom_of_reno}. Thus a necessary
condition for the string variation of expression~\eqref{eq:TL1mL1m:a}
to satisfy Eq.~\eqref{eq:SI2} is to impose the form \eqref{eq:p} as a
renormalization condition on $\vev{T_rhh'}$.

In summary, the string variation of the first term in
\eqref{eq:TL1mL1m:a} is a total divergence and we find:
\begin{align*}
\dc\Big(\underbracket{L_{1,\textup{mat}}L}\!_{1,\textup{mat}}\Big)
&\modd - \sfrac{1}{2} h'_{\rho\sigma} w_\nu\,\del_\mu
\sum_{\varphi,\chi'} \frac{\partial \Th^{\mu\nu}_\textup{mat}}{\partial \varphi}
\vev{\Tren \varphi \chi'} \frac{\partial {\Th'}^{\rho\sigma}_\textup{mat}}{\partial
\chi'} +(x\lra x').
\end{align*}
That is to say, obstructions to second-order SI at tree level can only
be caused by a violation of the \textbf{Ward identity} associated to
the (point-localized) matter stress-energy tensors, that is:
\begin{align}
\label{eq:violWard}
\del_\mu \Tren\Th^{\mu\nu}_\textup{mat}
{\Th'}^{\rho\sigma}_\textup{mat}\Big\vert_{\rm tree}=\pa_\mu \sum_{\varphi,\chi'} \frac{\partial
  \Th^{\mu\nu}_\textup{mat} }{\partial \varphi} \vev{\Tren
\varphi \chi'} \frac{\partial
{\Th'}^{\rho\sigma}_\textup{mat}}{\partial\chi'} \neq0.
\end{align}

We compute the obstructions to string-independence at second order
tree level corresponding to the possible violation of the Ward
identities for the following stress-energy tensors:
\begin{align}
\Th^{\mu\nu}_\phi &= \del^\mu \phi \del^\nu \phi - \tfrac{1}{2}
\eta^{\mu\nu} \big( \del_\ka \phi \del^\ka \phi -m^2 \phi^2 \big)
&&\text{of a scalar field }\phi(x);
\label{eq:Theta_scalar} 
\\
\Th^{\mu\nu}_F &= - F^{\mu\ka} {F^\nu}_\ka+\tfrac{1}{4} \eta^{\mu\nu}
F^{\ka\la} F_{\ka\la} &&\text{of the Maxwell tensor}~F_{\mu\nu}(x);
\label{eq:Theta_Maxwell} 
\\
\Th^{\mu\nu}_\psi &= \tfrac{i}{4} \Bigl(\overline{\psi}\gamma^\mu
\overset{\leftrightarrow}{\partial^\nu} \psi + \overline{\psi}
\gamma^\nu \overset{\leftrightarrow}{\partial^\mu}\psi \Bigr)
&&\text{of a Dirac field } \psi(x).
\label{eq:Theta_fermi} 
\end{align}
Notice first that the trace parts of the stress-energy tensors do not
contribute to the tree graph, because $h_{\mu\nu}$ is traceless and
the renormalizations are determined such that
\begin{align*}
\eta^{\mu\nu} \vev{\Tren h_{\mu\nu} h'_{\rho\sigma}} = 0 =
\eta^{\rho\sigma} \vev{\Tren h_{\mu\nu} h'_{\rho\sigma}}.
\end{align*}

We are going to show:
\begin{proposition}
\label{prop:matter_obst}
For the matter couplings $L_{1,\rm mat} = \frac12 h_{\mu\nu}
\Theta^{\mu\nu}_{\rm mat}$ with either of
\eqref{eq:Theta_scalar}--\eqref{eq:Theta_fermi}, there is a
non-resolvable obstruction of the interaction $L_{1,\rm mat}$\,, of
the universal form:
\begin{align}
\label{eq:obst_univ}
\sO_{2,\rm mat}(x,x') = -i\,\Th^{\mu\nu}_{\rm mat} \, w^\ka \big(
\del_\ka h_{\mu\nu} - \del_\mu h_{\ka\nu} - \del_\nu h_{\ka\mu}\big)\,
\delta(x-x').
\end{align}
\end{proposition}
The proof proceeds by a case-by-case analysis in the next three
subsections. The cases are mutually dissimilar enough that one is
allowed to conjecture universality of the above form.

\subsubsection{Graviton coupling to scalar fields}
\label{ssec:scalar}

Having seen that the contractions $\vev{\Tren hh'}$ in
$\underbracket{L_{1,\phi}L'}\!\!_{1,\phi}$ do not contribute to the
obstruction, we study the contractions of the scalar fields. Since the
string-localized graviton potential is traceless, in view of
Eq.~\eqref{eq:TL1mL1m:a} those are:
\begin{align}
\label{eq:scalar_tree_phi_cont}
\sfrac{1}{4} h_{\mu\nu}h'_{\rho\sigma}  \sum_{\varphi,\chi'} 
\frac{\partial \Th^{\mu\nu}_\phi}{\partial \varphi} \vev{\Tren \varphi \chi'}
\frac{\partial {\Th'_\phi}^{\rho\sigma}}{\partial\chi'}
= h_{\mu\nu}\del^\mu \phi \vev{\Tren\,\del^\nu
\phi\,{\del'}^\sigma\phi'}h'_{\rho\sigma}{\del'}^\rho \phi'.
\end{align}
Because the scaling dimension of the renormalized propagator on the
right-hand side of \eqref{eq:scalar_tree_phi_cont} equals the
spacetime dimension, it admits one free renormalization parameter:
\begin{align}
\label{eq:scalarreno}
\vev{\Tren \del^\nu \phi \, {\del'}^\sigma \phi' } &= \vev{T_0\,
\del^\nu \phi\,{\del'}^\sigma \phi'} + \vev{T_r \del^\nu \phi\,
{\del'}^\sigma \phi'}
\\ 
\notag &= \vev{T_0 \del^\nu \phi {\del'}^\sigma \phi' } + ic_\phi
\eta^{\nu\sigma}\delta(x-x').
\end{align}
The non-kinematic part of \eqref{eq:scalar_tree_phi_cont} can
always be absorbed into an \textit{induced interaction}:
\begin{align*}
L_{2,r,\phi} := c_\phi (hh)^{\mu\nu} \del_\mu \phi {\del}_\nu \phi =
c_\phi \big(\doublcont{hh\Th_\phi} + \doublcont{hh} L_{0,\phi} \big),
\end{align*}
where $L_{0,\phi}$ is the quadratic (free) scalar ``Lagrangian''
visible in \eqref{eq:Theta_scalar}, so that
\begin{align*}
i L_{2,r,\phi} \delta(x-x') +i^2 h_{\mu\nu} \del^\mu \phi\,\vev{T_r
\del^\nu \phi {\del'}^\sigma \phi'}\, h'_{\rho\sigma} {\del'}^\rho
\phi' = 0.
\end{align*}

In conclusion, \textit{only the kinematic part} can contribute a
possible obstruction to SI. Since one has for it:
\begin{align*}
\del_\mu \vev{T_0 \del^\mu \phi {\del'}^\sigma \phi' } = \square
{\del'}^\sigma D_{m,F}(x-x') = -m^2 \vev{T_0\,\phi {\del'}^\sigma
\phi'} - i{\del'}^\sigma\delta(x-x'),
\end{align*}
we find
\begin{align*}
&\dc\Big(\underbracket{L_{1,\phi}L'}\!\!_{1,\phi}\big\vert_{T_0}\Big)
= \delta_c \big( h_{\mu\nu} \del^\mu \phi \, \vev{T_0 \del^\nu \phi
{\del'}^\sigma \phi'}\,h'_{\rho\sigma} {\del'}^\rho \phi' \big)
\\
&\modd w_\mu \del^\mu \phi \, h'_{\rho\sigma} {\del'}^\rho
\phi'{\del'}^\sigma\cdot i\delta(x-x') + (x\lra x') \modd: i\delta_c
\big(L_{2,0,\phi}\big)\delta(x-x') + \sO_{2,\phi}(x,x'),
\end{align*}
with the \textit{induced interaction} density:
\begin{align*}
L_{2,0,\phi} = (hh)_{\mu\nu} \Th_\phi^{\mu\nu} + \sfrac{1}{4}
\doublcont{hh} \big((\del\phi\del\phi) - m^2\phi^2\big) =
\doublcont{hh\Th_\phi} + \sfrac{1}{2} \doublcont{hh} L_{0,\phi},
\end{align*}
plus the remaining \textit{obstruction}:
\begin{align}
\label{eq:obst_scalar}
\sO_{2,\phi}(x,x') = \Th_\phi^{\mu\nu} w^\ka \big(\del_\mu h_{\ka\nu}
+ \del_\nu h_{\ka\mu} - i\del_\ka h_{\mu\nu} \big)\delta(x-x'),
\end{align}
which cannot be resolved by adding an induced term with field content 
$(\phi,\phi, h ,h)$. In total, we collect: 
\begin{align*}
\dc\Bigl(\underbracket{L_{1,\phi}L'}\!\!_{1,\phi}\Bigr) \modd \delta_c
\bigl( (1+c_\phi) \doublcont{hh \Th_\phi} + i\big(\sfrac12 + c_\phi\bigr)
\doublcont{hh} L_{0,\phi} \bigr)\delta(x-x') + \sO_{2,\phi}(x,x').
\end{align*}

\subsubsection{Graviton coupling to electromagnetic fields}
\label{ssec:Maxwell}

Keeping in mind that that contractions $\vev{\Tren hh'}$ in
$\underbracket{L_{1,\phi}L'}\!\!_{1,\phi}$ did not contribute to the
obstruction, we now study the contractions for electromagnetic fields:
\begin{align}
\label{eq:Maxwell_tree_FF_cont}
&\quad\sfrac{1}{4} h_{\mu\nu} h'_{\rho\sigma} \sum_{\varphi,\chi'}
\frac{\partial \Th^{\mu\nu}_F }{\partial \varphi} \vev{\Tren \varphi\chi'}
\frac{\partial{\Th'_F}^{\rho\sigma}}{\partial\chi'} = h_{\mu\nu}{F^{\mu}}_{\ka} \vev{\Tren
F^{\nu\ka} {F'}^{\sigma\la}}h'_{\rho\sigma} {{F'}^{\rho}}_{\la}.
\end{align}
Similarly to the scalar case, the renormalized propagator of the Maxwell 
field has an ambiguity:
\begin{align*}
&\vev{\Tren F^{\nu\ka} {F'}^{\sigma\la}} = \vev{T_0 F^{\nu\ka}
{F'}^{\sigma\la}} + \vev{T_r F^{\nu\ka} {F'}^{\sigma\la}}
\\
&= \vev{T_0 F^{\nu\ka} {F'}^{\sigma\la}} +i c_F\bigl(\eta^{\nu\sigma}
\eta^{\ka\la} - \eta^{\nu\la} \eta^{\ka\sigma}\bigr)\delta(x-x'),
\end{align*}
with a free parameter $c_F$. As before, the non-kinematic part can 
be absorbed into an induced interaction
\begin{align}
L_{2,r,F} :=&\; c_F \big( (hh)^{\rho}_\mu {F^{\mu\ka}} F_{\rho\ka} -
h_{\mu\nu} h_{\rho\sigma} {F^{\mu\sigma}} {F}^{\rho\nu}\big) \\ \notag
=&\; -c_F \big( \doublcont{hh\Th_F} + \doublcont{hh} L_{0,F} +
h_{\mu\nu} h_{\rho\sigma} {F^{\mu\sigma}} {F}^{\rho\nu} \big),
\end{align}
where $L_{0,F}$ denotes the quadratic (free) Maxwell Lagrangian, so that
\begin{align}
i L_{2,r,F} \delta(x-x') + i^2 h_{\mu\nu} {F^{\mu}}_{\ka} \vev{T_r
F^{\nu\ka} {F'}^{\sigma\la}} h'_{\rho\sigma} {{F'}^{\rho}}_{\la} = 0.
\end{align}
Hence, like in the scalar case, the only possible obstruction to string 
independence can come from the kinematic part in 
\eqref{eq:Maxwell_tree_FF_cont}. We have:
\begin{align}
\label{eq:div_T0FF}
\del_{\mu} \vev{T_0 F^{\mu\nu} {F'}^{\ka\la}} = \big( \eta^{\nu\la}
{\del'}^\ka - \eta^{\nu\ka} {\del'}^\la \big) i\delta(x-x') = -i \big(
\eta^{\nu\la} \del^\ka - \eta^{\nu\ka} \del^\la \big)\delta(x-x').
\end{align}
A rather lengthy computation, requiring use of the Bianchi identity
for the Maxwell field and~\eqref{eq:div_T0FF} -- cf.~\cite{GPhD} in
this repect -- yields
\begin{align*}
\dc\Bigl(\underbracket{L_{1,F}L'}\!\!_{1,F}\big\vert_{T_0}\Bigr) =
i\delta_c \big(h_{\mu\nu} h'_{\rho\sigma}\big) {F^{\mu}}_{\ka}
{{F'}^{\rho}}_{\la} \vev{T_0 F^{\nu\ka} {F'}^{\sigma\la}} \modd
\delta_c \big(L_{2,0,F}\big)\delta(x-x')+ \sO_{2,F}(x,x'),
\end{align*}
with the induced interaction density: 
\begin{align}
\notag
L_{2,0,F} &= - F_{\nu\la} {F_\mu}^{\la} (hh)^{\mu\nu} - \sfrac{1}{2}
h_{\mu\nu} {h}_{\rho\sigma} F^{\mu\rho} F^{\nu\sigma} + \sfrac{1}{8 }
F_{\nu\la} {F}^{\nu\la} \doublcont{hh}
\\
&= \doublcont{hh \Th_F} + \sfrac{1}{2} \doublcont{hh} L_{0,F} -
\sfrac{1}{2} h_{\mu\nu} {h}_{\rho\sigma} F^{\mu\rho} F^{\nu\sigma}
\end{align}
and the remaining obstruction: 
\begin{align}
\label{eq:obst_Maxwell}
\sO_{2,F}(x,x') = -i\Th^{\mu\nu}_F w^\ka \big( \del_\ka h_{\mu\nu} -
\del_\mu h_{\ka\nu} - \del_\nu h_{\ka\mu} \big)\delta(x-x'),
\end{align}
totally analogous to the one in the scalar case. As in the latter, the
obstruction $\sO_{2,F}$ cannot be resolved by an induced interaction
with field content $(F,F,h,h)$. In all, one finds:
\begin{align*}
\dc\bigl(\underbracket{L_{1,F}L'}\!\!_{1,F}\bigr) =&\; \dc \bigl(
(1-c_F) \doublcont{hh\Th_F} + \bigl(\sfrac12 - c_F\bigr) \doublcont{hh}
L_{0,F}
\\
&\; -i\bigl(\sfrac12 + c_F\bigr) F_{\nu\rho} F_{\la\sigma} h^{\nu\la}
{h}^{\rho\sigma} \bigr)\delta(x-x') + \sO_{2,F}(x,x').
\end{align*}

\subsubsection{Graviton coupling to Dirac fields}
\label{ssec:Dirac}

The Fermi case is more involved than the scalar case and, differently
from the previous Maxwell case, was not treated in \cite{GPhD}.
For these reasons we give a slightly more detailed treatment here.

Since, as in the other cases, the contraction $\vev{\Tren hh'}$ does not
contribute to the obstruction, we study only the contractions of the
Dirac fields. As before, the obstruction is caused by the violation of the Ward
identity corresponding to the fermion stress-energy tensor under the
kinematic part of the time-ordered product,
\begin{align}
\label{eq:dcLpsiLpsi-1}  \notag
  \dc\Bigl(\underbracket{L_{1,\psi}L'}{}\!\!_{1,\psi}\big\vert_{T_0}\Bigr)&=\delta_c\Biggl( \sfrac{1}{4} h_{\mu\nu}h'_{\rho\sg} 
\sum_{\varphi,\chi'} \frac{\partial \Th^{\mu\nu}_\psi }{\partial \varphi} \vev{T_0
\varphi \chi'} \frac{\partial {\Th'_\psi}^{\rho\sg}}{\partial \chi'}  \Biggr)
\\
&\modd -\sfrac{1}{2} w_{\nu} h'_{\rho\sg}\del_\mu
\sum_{\varphi,\chi'} \frac{\partial \Th^{\mu\nu}_\psi }{\partial \varphi} \vev{T_0
\varphi \chi'} \frac{\partial {\Th'_\psi}^{\rho\sg}}{\partial \chi'} +(x\lra x').
\end{align}
A somewhat protracted calculation yields: 
\begin{align*}
\del_\mu \sum_{\varphi,\chi'} \frac{d \Th^{\mu\nu}_\psi }{d \varphi}
\vev{T_0 \varphi \chi'} \frac{d {\Th'_\psi}^{\rho\sg}}{d \chi'} =&
\,\sfrac{1}{16} \Bigl[\overline{\psi} \Bigl(-\lrd_\nu \gamma_\rho + i
\gamma_\nu (i \del^\mu \gamma_\mu +m) \gamma_\rho \Bigr) \delta(x-x')
\lrd{}'_\sg \psi'
\\
&+ \overline{\psi}' \lrd{}'_\sg \delta(x-x') \Bigl(\gamma_\rho
\lrd_\nu + \gamma_\rho i\ld_\mu \gamma^\mu -m) i\gamma_\nu \Bigr)\psi
\Bigr] + (\rho \lra \sg).
\end{align*}
After several integrations by parts, and by use of the equations of
motions of the fields, the string variation \eqref{eq:dcLpsiLpsi-1} is found to be:
\begin{align}
\label{eq:dcLpsiLpsi-2}  \notag
\dc\Bigl(\underbracket{L_{1,\psi}L'}{}\!\!_{1,\psi}\big\vert_{T_0}\Bigr)
&\;\modd \Big[ w^\nu \del_\nu{h}^{\rho\sg} \Big(
\overline{\psi}\gamma_\rho \del_\sg \psi - \del_\sg
\overline{\psi}\gamma_\rho \psi \Big) + w^\nu {h}^{\rho\sg} \Big(
\overline{\psi}\gamma_\rho \del_\sg \del_\nu \psi - \del_\sg \del_\nu
\overline{\psi}\gamma_\rho \psi \Big) 
\\ \notag
&\quad +\sfrac{1}{4} \del^\mu w^\nu {h}^{\rho\sg} \Big(
\overline{\psi} \gamma_\mu \gamma_\nu \gamma_\rho \del_\sg \psi
+\del_\sg \overline{\psi} \gamma_\mu \gamma_\nu \gamma_\rho \psi \Big)
\\
&\quad +\sfrac{1}{2} \del_\rho w^\nu {h}^{\rho\sg} \del_\sg
\overline{\psi} \gamma_\nu \psi -\sfrac{1}{2} \del^\mu w_\rho
{h}^{\rho\sg} \del_\sg \overline{\psi} \gamma_\mu \psi
\Big]\,\delta(x-x').
\end{align}
Due to the many identities satisfied by the Dirac fields, the
$\gamma$-matrices and the string-localized graviton field, there are
just four linearly independent (up to total divergences) hermitean candidates
for induced interactions with the field and derivative content
$(\del,\overline{\psi},\psi,h, h)$:
\begin{align}
  \label{eq:L2n_fermi}
  L_{2,\psi}^1 &:= (hh)_{\mu\nu} \Th_\psi^{\mu\nu},
&&L_{2,\psi}^2 := ( h h)^{\rho\sg} \del_\sg j_\rho, 
\\ \notag
L_{2,\psi}^3 &:= \sfrac{i}{2} \doublcont{hh} \overline{\psi}
\gamma^\mu \lrd_\mu \psi = \doublcont{hh} \Th_{\mu,\psi}^{\mu},
&&L_{2,\psi}^4 := \sfrac i2(h\del^\mu h)^{\rho\sg} \overline{\psi} (\gamma_\mu
\gamma_\rho \gamma_\sg-\gamma_\sg\gamma_\rho \gamma_\mu
) \psi.
\end{align}
Yet another ponderous calculation allows to split \eqref{eq:dcLpsiLpsi-2} as
\begin{align*}
\dc\Bigl(\underbracket{L_{1,\psi}L'}{}\!\!_{1,\psi}\big\vert_{T_0}\Bigr)
&\modd i\delta_c \big(L_{2,0,\psi}\big)\,\delta(x-x')+
\sO_{2,\psi}(x,x')
\end{align*}
with the induced interaction density:
\begin{align}
  \label{eq:L20_fermi}
L_{2,0,\psi} &= \sfrac{3}{4}  L_{2,\psi}^1-\sfrac
               1{8} L_{2,\psi}^4,
\end{align}
and the remaining non-resolvable obstruction:
\begin{align}
\label{eq:obst_fermi}
\sO_{2,\psi}(x,x') = -i\Th^{\mu\nu}_\psi w^\ka \big(\del_\ka
h_{\mu\nu} - \del_\mu h_{\ka\nu} - \del_\nu h_{\ka\mu} \big)
\delta(x-x').
\end{align}

\begin{remark}
Similarly to the scalar and Maxwell cases, some of the fermion
propagators have ambiguities. The ambiguous propagators are
$\vev{\Tren \del_\mu \overline{\psi} \psi'}$, $\vev{ \Tren
\overline{\psi} \del'_\mu\psi'} $ and $\vev{\Tren \del_\mu
\overline{\psi} \del'_\nu\psi'}$. Each of the first two has one free
parameter corresponding to a correction $\gamma_\mu \delta(x-x')$. The
third one has several free parameters, corresponding to corrections
proportional to $\delta(x-x')$ and $\del_\rho\delta(x-x')$, with
several prefactors exhibiting combinations of $\gamma_\rho$ and
$\eta_{\rho\sg}$. Since all these corrections correspond to induced
interactions, they cannot resolve the obstruction
\eqref{eq:obst_fermi}. Therefore they are secondary to the present
issue, and thus we omit further detail.
\end{remark}

This concludes the proof of Prop.~\ref{prop:matter_obst}.
\hfill$\square$

\subsection{Graviton self-coupling}
\label{sec:Selfcoupling}

A cubic self-coupling of massless string-localized fields of
helicity~2 without derivatives cannot satisfy the $L$-$Q$-pair condition
for a string-independent action~\cite{GassMThesis}. The next
possibility to consider is cubic self-couplings with two derivatives.
\textit{A priori} and dropping all divergences, the general form of
$L_1$ may consist of only three terms:
\begin{align*}
h^{\mu\nu}(\al_1\del_\mu h^{\rho\sg} \del_\nu h_{\rho\sg} + \al_2
\del^\rho h^\sg_\mu\del_\sg h_{\nu\rho} + \al_3\,
\del^\rho h^\sg_\mu\del_\nu h_{\sg\rho}).
\end{align*}
In fact Nature chooses $\al_1=\frac12, \al_2=1,\al_3=0$: a cubic
self-coupling of massless string-localized fields of helicity 2, which
satisfies the $L$-$Q$-pair condition, is \textit{unique} up to a total
divergence and a multiplicative constant \cite{GassMThesis}. 

Since there is the freedom of adding divergences, we will take the
liberty to subtract the term $h^{\mu\nu}\del^\rho h_{\sg\mu}\del_\rho
h^\sg_\nu$, because it entails Eq.~\eqref{eq:dU} below, which is going
to simplify computations, including the renormalizations. In fact, by
adding divergences $L_1$ can be given many other forms with
complementary conveniences. Our choice~\eqref{eq:L1} for $L_1$ below
\textit{is identical} with the classical expansion of the Einstein
action -- which also has a certain freedom of adding divergences
\cite[Chap.\ 5]{ScharfLast} -- when the classical metric deviation
field $h_{\mu\nu}$ is replaced by the string-localized quantum field
$h_{\mu\nu}(x;c)$ of helicity 2 -- see \aref{a:Expansion}.

For a better grasp of the contraction scheme involved in the present
problem, on regarding $(\pa_\mu h)_{\ka\la}= \pa_\mu h_{\ka\la}$ as a
symmetric tensor, we rewrite $L_1$ in the notation of Convention
\ref{def:notation} as:
\begin{align}
\label{eq:L1}
L_1 &= h^{\mu\nu}\bigl[\sfrac12 \doublcont{\del_\mu h \del_\nu h} 
+ \del^\ka h_{\mu\la} \del^\la h_{\nu\ka}  
- (\pa_\ka h \pa^\ka h)_{\mu\nu} \bigr].
\end{align}
Turning to perturbation theory, let us write the Wick expansion as:
\begin{align}
\label{eq:Wick}
\underbracket{L_1\chi'\!} =
V^{\mu\nu}\vev{\Tren h_{\mu\nu}\chi'}+W^{\ka,\mu\nu}
\vev{T_{\rm ren}\pa_\ka h_{\mu\nu}\chi'},
\end{align}
where we have introduced the quadratic fields
\begin{align*}
V^{\mu\nu}&:= \sfrac12\bb{\pa^\mu h\pa^\nu h} + \pa_\ka
h^{\mu\la} \pa_\la h^{\nu\ka} - (\pa_\ka h \pa^\ka h)^{\mu\nu},
\\
W^{\ka,\mu\nu} &:= h^{\ka\la}\pa_\la h^{\mu\nu} + (h\pa^\nu h
)^{\mu\ka} + (\pa^\mu h h)^{\ka\nu} - \pa^\ka(hh)^{\mu\nu}.
\end{align*}
We have to separately renormalize $\vev{\Tren hh'}$ and $\vev{\Tren
\pa hh'}$, so as to respect the trace conditions. It turns out that
the SI condition can be fulfilled with
\begin{align}
\label{eq:Td=dT}
\vev{T_{\rm ren}\pa_\ka h_{\mu\nu}\chi'}=\pa_\ka \vev{T_{\rm ren}
h_{\mu\nu}\chi'}
\end{align}
both for $\chi'=h'$ and $\chi'=\pa 'h'$. Thus, with an integration by
parts,
\begin{align}
\label{eq:UVW} 
\underbracket{L_1\chi}\!'\modd U^{\mu\nu}\vev{\Tren h_{\mu\nu}\chi'}, 
\quad
\word{with} U^{\mu\nu}: = V^{\mu\nu}-\pa_\ka W^{\ka,\mu\nu},
\end{align}
and in particular:
\begin{align}
\label{eq:Wickd2}
\underbracket{L_1L}\!'_1 \modd
U^{\mu\nu}\vev{\Tren h_{\mu\nu}h'_{\rho\sg}}U'^{\rho\sg}.
\end{align}
The field $U$ satisfies:
\begin{align}
\label{eq:dU}
\pa_\mu U^{\mu\nu} = \pa^\nu S \quad\word{with} S =
\sfrac18\square\bb{hh}- \sfrac12\pa_\ka\pa_\la(hh)^{\ka\la},
\end{align}
and there holds the identity 
\begin{align}\label{eq:US}
  U_\mu^\mu -2S = \sfrac12 \square\bb{hh}.
\end{align}

Now we can compute the obstruction of the string-localized graviton
self-interaction, exploiting features of the kinematic and
renormalized propagators. The latter are expressed with the help of
string integrated differential operators $a_{\mu,\rho\sg}$,
$b_{\rho,\sg}$, $c_\rho$ acting on the Feynman propagator, which are
displayed in \aref{a:Propagators}. The method is a ``cascade of
integrations by parts'', illustrating the power of the feature of SQFT
that string-dependent parts are derivatives.

We begin with the kinematic propagator, of which we separate in the
first step the string-independent part:
\begin{align}
\vev{T_{0,\ast}h_{\mu\nu}h'_{\rho\sg}} =
\sfrac12\big[\eta_{\mu\rho}\eta_{\nu\sg}+\eta_{\mu\sg}\eta_{\nu\rho}-
\eta_{\mu\nu}\eta_{\rho\sg}\big]D_{0,F}(x-x').
\end{align}
It does coincide with the propagator that one would use in gauge
theory, which is derived from an indefinite two-point function, and
is not traceless.%
\footnote{In the gauge-theoretic setting there is no need for
renormalization of the propagators, because the field cannot be set to
zero on the indefinite metric Fock space; so there are no trace
conditions.}

The string-dependent part is of the form \eqref{eq:T0chh}:
\begin{align}
\vev{T_{0,c}h_{\mu\nu}
h'_{\rho\sg}}&=\sfrac12\big[\pa_{(\mu}a_{\nu),\rho\sg}+
a'_{\mu\nu,(\rho}\pa_{\sg)}\big] D_{0,F}(x-x'),
\end{align}
where $\pa_{(\mu}a_{\nu),\rho\sg}$ stands for the symmetric sum
$\pa_{\mu}a_{\nu,\rho\sg} + \pa_{\nu}a_{\mu,\rho\sg}$. The
string-integrated differential operators $a_{\nu,\rho\sg}$ and
$a_{\mu\nu,\rho}$ are detailed in \eqref{eq:A} and \eqref{eq:A'}. This
gives:
\begin{align}
\notag 
&T_0\underbracket{L_1L}\!'_1 \modd U^{\mu\nu}
\vev{T_{0,\ast}h_{\mu\nu}h'_{\rho\sg}}
U'^{\rho\sg}+U^{\mu\nu} \vev{T_{0,c}h_{\mu\nu}h'_{\rho\sg}}
U'^{\rho\sg}
\\
\label{self_kin}
&= U^{\mu\nu} D_{0,F} U'_{\mu\nu} -\sfrac12U_\mu^\mu D_{0,F}
U'^\rho_\rho +\sfrac12 U^{\mu\nu}
\bigl(\pa_{(\mu}a_{\nu),\rho\sg}+a'_{\mu\nu,(\rho}\pa_{\sg)}\bigr)
D_{0,F} U'^{\rho\sg}.
\end{align}
Of the last term, we compute the contribution of
$\pa_{(\mu}a_{\nu),\rho\sg}$:
\begin{align}
\label{eq:casc}
U^{\mu\nu} \sfrac12\pa_{(\mu}a_{\nu),\rho\sg} D_{0,F} U'^{\rho\sg}
\modd& -\pa_\mu U^{\mu\nu} a_{\nu,\rho\sg} D_{0,F} U'^{\rho\sg} =
-\pa^\nu S a_{\nu,\rho\sg}D_{0,F} U'^{\rho\sg}
\\
\notag 
\modd & S  a_{\nu,\rho\sg}\pa^\nu D_{0,F}
U'^{\rho\sg}= S (\eta_{\rho\sg} + b_{\rho\sg} \square +
c_{(\rho}\pa_{\sg)}) D_{0,F} U'^{\rho\sg}
\\ 
\notag 
\modd & S D_{0,F} U'^\rho_\rho - iS b_{\rho\sg}\delta(x-x')
U'^{\rho\sg}+2S c_{\rho} D_{0,F} \pa'^\rho S'
\\
\notag 
\modd & S D_{0,F} U'^\rho_\rho - iSb_{\rho\sg}\delta(x-x')
U'^{\rho\sg}+2S c_{\rho}\pa^\rho D_{0,F} S'.
\end{align}
In the second line, we have used the identity \eqref{eq:abc}.
Likewise, we now use from \eqref{eq:bc}
\begin{align*}
c_\rho \pa^\rho = -1 + \sfrac14 \big(I'^2-(II')\big)\square.
\end{align*}
Adding the corresponding contribution from
$a'_{\mu\nu,(\rho}\pa_{\sg)}$ in \eqref{self_kin}, we get
\begin{align}
\label{eq:TL1L1}
&T_0\underbracket{L_1L}\!'_1 \modd U_{\rho\sg} D_{0,F} U'^{\rho\sg}
-\sfrac12(U_\mu^\mu -2S) D_{0,F} (U'^\rho_\rho-2S') -2S D_{0,F}S'
\\
\label{eq:TL1L1_string}
&- iS b_{\rho\sg} \delta(x-x') U'^{\rho\sg} - iU^{\mu\nu}
b'_{\mu\nu}\delta(x-x') S'- \frac i2 S (I-I')^2\delta(x-x')S'.
\end{align}
The second line \eqref{eq:TL1L1_string} contains string-integrated
$\delta$ functions -- which we must \textit{not have} by condition
\eqref{eq:SI2}. Lemma \ref{lem:cancel_self} in \aref{a:Propagators}
proves that the renormalization part of the propagator, that is
$T_r\underbracket{L_1L}\!'_1$, exactly cancels
\eqref{eq:TL1L1_string}. Thus the obstruction of the self-interaction
is the string variation of the first line~\eqref{eq:TL1L1}.

The string variations of the fields $U$ and $S$ are computed to be:
\begin{align}
\label{eq:deltaU}
&\delta_c\big(U^{\mu\nu}\big) = \square\big(2(\pa^\mu w_\ka+\pa_\ka
w^\mu)h^{\ka\nu}) - \square \big(w_\ka(\pa^\ka h^{\mu\nu} -\pa^\mu
h^{\ka\nu}- \pa^\nu h^{\mu\ka})\big) + \pa^{(\mu}K^{\nu)};
\\
\label{eq:deltaS} 
&\delta_c\big(S\big) = \sfrac14 \delta_c \big(\square{\bb{hh}}\big) +
\pa_\mu K^\mu;\!\!\!\word{with}
\\
\label{eq:K}
&K^\mu = -\square(w_\ka h^{\ka\mu}) - h^{\ka\la}\pa_\ka\pa_\la w^\mu,
\quad \pa_\mu K^\mu = \delta_c\big(S-\sfrac14 \square\bb{hh}\big).
\end{align}
When the string variation of \eqref{eq:TL1L1} is computed, there arise 
$K$-terms not containing a wave operator, hence producing undesired bulk
contributions. Happily, these bulk terms cancel each other:
\begin{align}
\notag
& \pa_{(\mu}K_{\nu)} D_{0,F} U'^{\mu\nu} +(x\lra x')\modd -2K_\nu
D_{0,F} \pa'^\nu S'+(x\lra x')
\\
\label{eq:KDU}
&\modd 2\pa_\nu K^\nu D_{0,F} S' + (x\lra x') = 2\delta_c\big(S
D_{0,F} S'\big) + iS\delta_c\big(\bb{hh}\big)\delta(x-x'),
\end{align}
by \eqref{eq:K}. Thus, the obstruction $\delta_c\big(T_{\rm
ren}\underbracket{L_1L}\!'_1\big)$ of the self-interaction is (up to
divergences)
\begin{align}
\notag
&-2\bigl(\delta_c\big((hh)^{\mu\nu}\big) - w_\ka(\pa^\ka h^{\mu\nu} -
\pa^\mu h^{\ka\nu}- \pa^\nu h^{\mu\ka})\big) U_{\mu\nu} +
\delta_c\big(\bb{hh}\big) S
\\
&\quad +\sfrac1{8}\delta_c\big(\bb{hh}\square\bb{hh}\bigr),
\label{eq:self_obst}
\end{align}
multiplying $i\delta(x-x')$. 

To establish SI at second order, it remains to show that
\eqref{eq:self_obst} is the string variation of a quartic induced
interaction $L_2$:
\begin{align}
\label{eq:deltaL2}
\eqref{eq:self_obst} \stackrel!=\delta_c(L_2).
\end{align}
 \begin{theorem}
\label{thm:resol_self_obstr}
The second order obstruction of the graviton
self-interaction is resolved by the induced interaction
\begin{align}
\label{eq:L2_ind}
L_{2} \modd -2(\pa_\nu h h h \pa_\mu h)^{\mu\nu} -\sfrac12\bb{h\pa_\mu
h}\bb{h\pa^\mu h} - 2\bb{h\pa_\mu h \pa_\nu h}h^{\mu\nu};
\end{align}
where $L_2$ is uniquely determined up to total divergences.
\end{theorem}
\begin{proof}
After the long preparations leading to \eqref{eq:self_obst}, the proof
is routine: work out \eqref{eq:self_obst} with \eqref{eq:deltaU} and
\eqref{eq:deltaS}, make the most general Ansatz to solve
\eqref{eq:deltaL2}, and compare coefficients. This leads to a highly
overdetermined system of twenty-four equations for eight unknowns,
with a unique solution. One obtains a large number of structures
$O(w,\del h\del h\del h)$. To control linear dependences (up to
divergences) one first detaches all derivatives from the factor
$w^\ka$ by an integration by parts. The remaining linear dependences
arise from $\pa_\ka w^\ka=0$ implying $(w\pa)X\modd0$, and from
$\square w^\ka=0$ implying $w^\ka\square X\modd 0$. One gets seven
different contraction schemes of the tensors $h$, with various
placements of the derivatives. Our basis of choice is:
\begin{align*}
\begin{array}{lll}
w^\ka h_{\ka\la}\bb{hh}:\quad
&S^1= w^\ka\cdot  h_{\ka\la}\bb{\pa^\la \pa_\al h\pa^\al h} ,\quad
&S^2= w^\ka\cdot  \pa_\al h_{\ka\la} \bb{\pa^\al h\pa^\la h},
\\
w^\ka\cdot h_{\ka\la}(hh)^{\al\be}:\quad
&S^3 = w^\ka\cdot  h_{\ka\la}(\pa^\la \pa_\al h \pa_\be h)^{\al\be} ,\quad
&S^4 = w^\ka\cdot  \pa_\al h_{\ka\la} (\pa_\be h\pa^\la h)^{\al\be}, 
\\ 
&S^5 = w^\ka\cdot  \pa_\al h_{\ka\la} (\pa_\be \pa^\la hh)^{\al\be},\quad
&S^6 = w^\ka\cdot \pa_\al\pa_\be h_{\ka\la} (h\pa^\la h)^{\al\be} , 
\\
w^\ka\cdot (hh)_{\ka\la} h^{\al\be}:\quad
&S^7 = w^\ka\cdot  (h\pa_\al\pa_\be    h)_{\ka\la} \pa^\la h^{\al\be},\quad
&S^8 = w^\ka\cdot  (\pa_\al h\pa_\be h)_{\ka\la} \pa^\la h^{\al\be},
\\ 
&S^9= w^\ka\cdot  (\pa_\al\pa_\be hh)_{\ka\la} \pa^\la h^{\al\be},\quad
&S^{10} = w^\ka\cdot  (\pa_\al\pa^\la h \pa_\be h)_{\ka\la} h^{\al\be},
\\ 
&S^{11} = w^\ka\cdot  (\pa_\al\pa_\be\pa^\la h h)_{\ka\la} h^{\al\be},\quad
&S^{12} = w^\ka\cdot  (\pa^\la h\pa_\al\pa_\be h)_{\ka\la}  h^{\al\be},
\\ 
w^\ka\cdot (hhh)_{\ka\la}:\quad
&S^{13} = w^\ka\cdot  (h \pa_\al \pa^\la h \pa^\al h)_{\ka\la},\quad
&S^{14} = w^\ka\cdot (\pa^\la h\pa_\al h\pa^\al h)_{\ka\la},
\\ 
&S^{15} =  w^\ka\cdot (\pa_\al h\pa^\la h\pa^\al h)_{\ka\la},\quad
&S^{16} =  w^\ka\cdot (\pa_\al \pa^\la hh\pa^\al h)_{\ka\la},
\\
w^\ka\cdot \bb{\pa_\ka hhh}:\quad 
&S^{17}=  w^\ka\cdot \bb{\pa_\ka \pa_\al h\pa^\al hh},
\\
w^\ka\cdot \bb{\pa_\ka hh} h^{\al\be}:\quad 
&S^{18} =  w^\ka\cdot \bb{\pa_\ka h\pa_\al\pa_\be h} h^{\al\be}, \quad
&S^{19} =  w^\ka\cdot   \bb{\pa_\ka \pa_\al h \pa_\be h}h^{\al\be},
\\ 
&S^{20} =  w^\ka\cdot  \bb{\pa_\ka \pa_\al \pa_\be hh} h^{\al\be},
\\ 
w^\ka\cdot (\pa_\ka hhh)^{\al\be}:\quad 
&S^{21} = w^\ka\cdot (\pa_\ka h\pa_\be h\pa_\al h)^{\al\be}, \quad
&S^{22} = w^\ka\cdot (\pa_\ka h\pa_\al\pa_\be hh)^{\al\be} , 
\\
&S^{23} = w^\ka\cdot (\pa_\ka\pa_\be h\pa_\al hh)^{\al\be} , \quad
&S^{24} = w^\ka\cdot (\pa_\ka\pa_\al hh\pa_\be h)^{\al\be}. 
\end{array}
\end{align*}
In this basis,
\begin{align}
\notag
\eqref{eq:self_obst} = &-4[S^{1}-S^{5}-S^{6}+S^{9}+S^{10}+S^{11}+S^{12}
+S^{13}
\\
&+S^{14}+S^{15}+S^{16}+S^{17}+S^{21}+S^{22}+S^{23}+S^{24}]
\label{eq:obstr_self} 
\end{align}

On the right-hand side of \eqref{eq:deltaL2}, there are eight linearly
independent (up to divergences) candidate structures for the induced
interaction $O(h^4)$ with two derivatives:
\begin{align*}
\begin{aligned}
&L_2^1:=\doublcont{h \del^\mu h} \doublcont{h\del_\mu h},&\qquad&
L_2^2:= \doublcont{\del_\mu h \del_\nu h} (hh)^{\mu\nu}, 
\\ \notag
&L_2^3:=\doublcont{hh} (\del_\nu h\del_\mu h)^{\mu\nu}, &\qquad&
L_2^4:= \doublcont{h\del_\mu h \del_\nu h} h^{\mu\nu}, 
\\
&L_2^5:= \doublcont{h\pa_\mu h h \pa^\mu h }, &\qquad&
L_2^6:= (\pa_\nu hhh \pa_\mu h)^{\mu\nu},
\\
&L_2^7:= (h\pa_\mu h \pa_\nu h h)^{\mu\nu}, &\qquad&
L_2^8:= (h\pa_\nu h \pa_\mu h  h)^{\mu\nu}. 
\end{aligned}
\end{align*}
One computes
\begin{align*}
\delta_c\big(L_2^1\big) &\modd 8S^1,
\\  \notag
\delta_c\big(L_2^2\big) &\modd -2[S^1+S^2-2S^4-2S^5-4S^6+S^{18}+S^{19}],
\\  \notag
\delta_c\big(L_2^3\big) &\modd -4[S^1+2S^3+S^{18}+2S^{19}+S^{20}],
\\  \notag
\delta_c\big(L_2^4\big) &\modd 2[S^7+S^8+S^9+S^{10}+S^{11}+S^{12} + S^{17}],
\\  \notag
\delta_c\big(L_2^5\big) &\modd -8[S^{15}+S^{16}],
\\  \notag
\delta_c\big(L_2^6\big) &\modd-2[S^5+S^6+S^7+S^8-S^{13}-S^{14}-S^{15}
-S^{16} -S^{21}-S^{22}-S^{23}-S^{24}],
\\ \notag
\delta_c\big(L_2^7\big) &\modd 2[2S^3+S^4+S^5+2S^{10}+S^{11}+S^{12}+S^{13} 
-S^{21}-S^{22}],
\\  
\delta_c\big(L_2^8\big) &\modd2[S^3+S^4+S^5+S^6 +S^8+S^9+S^{10}+S^{11}+
S^{16} +S^{21}+S^{23}+2S^{24}].
\end{align*}
Comparing the coefficients of $S^a$ with the ones in
\eqref{eq:obstr_self}, the SI condition \eqref{eq:deltaL2} gives 24
equations for the eight coefficients of $L_2^n$ with the unique
solution:
\begin{align*}
L_2\modd -\sfrac12 L_2^1 -2L_2^4 - 2L_2^6 = L_2
\end{align*}
in~\eqref{eq:L2_ind}. This completes the proof. 
\end{proof}

As a byproduct of the analysis of string variations of interaction
structures with the field content $hhhh$ and two derivatives --
namely, uniqueness of the solution \eqref{eq:L2_ind} -- we have found:
\begin{corollary}
\label{cor:uniqueness}
Any interaction density with field and derivative content $(\del,\del,
h,h,h,h)$, which is string-independent on its own, must itself be a
total divergence.
\end{corollary} 

There is another side-message in the proof of Theorem
\ref{thm:resol_self_obstr}: the effectiveness and elegance of the
computations shows that SQFT -- technically complicated as it might
appear at first sight -- is hardly more difficult than causal
perturbation theory in gauge theory, as in \cite{ScharfLast}. One
diligently ought to know how to take advantage of systematic
cancellations before starting actual computations. By Lemmata
\ref{lem:cancel_self} and \ref{lem:cancel_matter} in
\aref{a:Propagators}, the necessary renormalization of propagators
only serves to cancel terms with string-dependent $\delta$-functions.

By inspection of \eqref{eq:L2_ind}, we conclude:
\begin{theorem}
\label{thm:reproduce_self}
The induced graviton self-interaction \eqref{eq:L2_ind} as a (Wick)
polynomial in the string-localized quantum field $h^{\mu\nu}(x;c)$
coincides with the second-order term of the classical Einstein action
\eqref{eq:eltrucodos} as a polynomial in the classical metric
deviation field $h^{\mu\nu}(x)$, defined by \eqref{eq:gold}.
\end{theorem}

Recall that the induced quantum coupling was determined only up to
divergences, and notice the last line in \eqref{eq:eltrucodos} is a
divergence. Recall as well that \eqref{eq:L1} was distinguished by the
$L$-$Q$-pair condition as a consistency condition based on quantum
principles. We thus are entitled to regard \eqref{eq:L1} and
\eqref{eq:L2_ind} as ``predictions'' from quantum theory on the
classical self-interaction for helicity 2.

\subsection{Cancellation of the matter obstructions}
\label{sec:Cancellation}

We have seen in Prop.\ \ref{prop:matter_obst} that the
matter obstruction to second-order string-independence in the
graviton-matter couplings is always of the form
\begin{align}
\label{obst_mat}
\dc\Big(\underbracket{L_{1,\rm mat}L}\!'_{1,\rm mat}\Big) \modd
i\delta_c \bigl( L_{2,0,\rm mat} + L_{2,r,\rm mat}\bigr) \delta(x-x')
+ \sO_{2,\rm mat}(x,x').
\end{align}
where $\sO_{2,\rm mat}$ is the non-resolvable
obstruction given in \eqref{eq:obst_univ}.

We will now show the second main result of this paper. To wit, that
these obstructions are \textit{exactly cancelled} if we take the
graviton self-interaction into account. Since the contractions of the
graviton field with matter fields do vanish, we do not need to specify
the concrete form of $\Th^{\mu\nu}_{\textup{mat}}$. We just assume
that it is conserved, point-localized (thus string-independent) and
symmetric. The cue is that \eqref{eq:obst_univ} features the
characteristic structure
\begin{align}
\sO_{2,\mu\nu} := w^\ka \big( \del_\ka h_{\mu\nu} - \del_\mu
h_{\ka\nu} - \del_\nu h_{\ka\mu}\big)
\end{align}
in the sector of helicity 2, which is not separately a string
variation. The same structure also appears in the string variation
\eqref{eq:deltaU} of the field $U_{\mu\nu}$.

Let us compute the interference term between the self-coupling
\eqref{eq:L1} and the matter coupling \eqref{eq:coup_matter} to an
arbitrary conserved and symmetric stress-energy tensor. We use again
\eqref{eq:Td=dT}, implying \eqref{eq:UVW}. Thus:
\begin{align}
\label{inter}
T_{\rm ren}\underbracket{L_1L}\!'_{1,\rm mat} &\modd \sfrac12
U^{\mu\nu} \vev{T_{\rm ren}h_{\mu\nu}h'_{\rho\sg}}
\Theta'^{\rho\sg}.
\end{align}
For the kinematic propagator of the form
\eqref{eq:T0*hh}+\eqref{eq:T0chh}, we can drop the contribution from
$a'_{\mu\nu,(\rho}\pa_{\sg)}$ because $\Theta$ is conserved. Therefore
\begin{align} 
\notag
T_0 \underbracket{L_1L}\!'_{1,\rm mat}&\modd \sfrac12\big[
(U_{\rho\sg} -\sfrac12 U_\mu^\mu\eta_{\rho\sg}) D_{0,F}
\Theta'^{\rho\sg} + \sfrac12 U^{\mu\nu} (\pa_{(\mu}a_{\nu),\rho\sg})
D_{0,F} \Theta'^{\rho\sg}\big]
\\ 
\notag &\modd\sfrac12\big[
(U_{\rho\sg} - \sfrac12 U_\mu^\mu\eta_{\rho\sg}) D_{0,F}
\Theta'^{\rho\sg} -\pa^\nu S a_{\nu,\rho\sg} D_{0,F}
\Theta'^{\rho\sg}\big] 
\\
\notag &\modd\sfrac12\big[
(U_{\rho\sg}-\sfrac12 U_\mu^\mu\eta_{\rho\sg}) D_{0,F}
\Theta'^{\rho\sg}+S (\eta_{\rho\sg}+b_{\rho\sg}\square)
D_{0,F} \Theta'^{\rho\sg}\big]
\\
\label{T0L1L1m}
&\modd\sfrac12\big[ (U_{\rho\sg}-(\sfrac12 U_\mu^\mu-S)\eta_{\rho\sg})
D_{0,F} \Theta'^{\rho\sg} - iS b_{\rho\sg}\delta(x-x')
\Theta'^{\rho\sg}\big].
\end{align}
The string variation of the first term in \eqref{T0L1L1m} is given by
\eqref{eq:deltaU} with \eqref{eq:US}, giving
\begin{align}
\label{obst_inter}
\delta_c\big(\sfrac12 (U_{\rho\sg}-(\sfrac12 U_\mu^\mu-S)\eta_{\rho\sg})
D_{0,F} \Theta'^{\rho\sg}\big)\hspace{50mm} \notag \\ \modd \frac i2 \big[
\big(\delta_c\big(\sfrac14\bb{hh}\eta_{\rho\sg}-(hh)_{\rho\sg}\big) +
O_{2,\rho\sg}\big) \Theta^{\rho\sg} \delta(x-x') \big].
\end{align}
Taking into account the term $x\lra x'$, we have to add twice the
above expression to the matter obstruction \eqref{obst_mat}. The
non-resolvable matter obstruction is cancelled, and the remaining
terms are cancelled by
\begin{align}\label{eq:L2_matter_ind}
L_{2,\rm mat} = \sfrac14\bb{hh}\bb{\Theta}-\bb{hh\Theta}+L_{2,0,\rm mat}+L_{2,r,\rm mat}.
\end{align}
There still remains the second term with string-integrated
$\delta$-functions in \eqref{T0L1L1m}. When we discard $\pa_\rho$,
$\pa_\sg$ in the string differential operators $b_{\rho\sg}$ given in
\eqref{eq:bc}, it becomes
\begin{align}
\label{eq:string}
-\frac i2 S \bigl(I_\rho I_\sg - \sfrac12 I^2\eta_{\rho\sg}\bigr)
\delta(x-x') \Theta'^{\rho\sg}.
\end{align}
Lemma \ref{lem:cancel_matter} in \aref{a:Propagators} asserts that this term is cancelled by the contribution of the
renormalization part of the propagators.
We conclude:
\begin{theorem}
\label{thm:matter_cancel}
The non-resolvable obstruction \eqref{eq:obst_univ} to string-independence 
at second order of the graviton-matter interaction is exactly cancelled 
if the graviton self-coupling is taken into account. The SI condition is satisfied with
\eqref{eq:L2_matter_ind}.
\end{theorem}

One may rephrase Theorem \ref{thm:matter_cancel} as follows: the
obstruction to SI arising from the violation of the Ward identity
\eqref{eq:violWard} is resolved if the graviton self-interaction is
taken into account.

Because the matter-matter obstructions for several free fields are
additive, and the same is true for the obstructions of the
interference term with the self-coupling, one has
\begin{corollary}
\label{cor:additivity}
The result of Theorem \ref{thm:matter_cancel} is true for the coupling
to any number
of scalar, Maxwell and Dirac fields, with the respective sum of
induced matter interactions.
\end{corollary}

It remains to compare \eqref{eq:L2_matter_ind} to the classical
expressions obtained from the expansion of the classical generally covariant
Lagrangians. This can be done ``by inspection'' of the
three cases. The expansions are sketched in \aref{a:Expansion}.

A perfect match with \eqref{eq:L2_matter_ind} is found in all three cases
if all matter renormalization constants are zero ($L_{2,r,\rm
  mat}=0$). Indeed, unlike for helicity~2, there is no physical reason to renormalize the
matter propagators with spin or helicity $\leq 1$. E.g., for the scalar field:
\begin{align}\label{}
L_{2,\phi} = \sfrac14\bb{hh}\bb{\Theta_\phi}-\bb{hh\Theta_\phi} +
\bb{hh\Theta_\phi} +\sfrac12\bb{hh}L_{0,\phi},
\end{align}
coincident with \eqref{eq:sc_exp}. 

We have thus established:
\begin{theorem}
\label{thm:reproduce_matter}
The induced matter interactions \eqref{eq:L2_matter_ind} as (Wick)
polynomials in the string-localized quantum field $h^{\mu\nu}(x;c)$ and
the quantum matter fields coincide, when the kinematical propagators
are chosen for the matter fields,
with the respective second-order terms of the classical generally
covariant scalar, Maxwell and Dirac Lagrangians as polynomials in the
classical metric deviation field $h_{\mu\nu}(x)$ and matter fields.
\end{theorem}

Theorem \ref{thm:matter_cancel} is what we called the lock-key
situation, in fact the central result of this paper. We emphasize
again that the self-interaction of gravitons knows nothing about
matter; the violation of the matter's Ward identity knows nothing
about the inner workings of gravity, either. Yet the obstructions are
resolved together, like a lock and its key. This is what we call
\textbf{quantum general covariance}. Theorems \ref{thm:reproduce_self}
and \ref{thm:reproduce_matter} are the precise formulation of the
assertion ``consistency of the quantum couplings predicts general
covariance'' in the Introduction.

\section{Conclusion and further discussion}
\label{sec:Conclusion}

For perturbative QFT of helicity 2 on a Minkowski background we have
established here a remarkable coincidence between the structure of
quantum interactions -- dictated by the quantum principles:
\emph{Hilbert space (positivity), covariance, locality of observables} -- and of
\emph{classical} ``Lagrangians'' -- dictated by classical general
covariance. It pertains not only to the self-interaction of particles
of spin 2 (``gravitons''), but also to matter couplings that go
necessarily through the stress-energy tensor.

Now, the quantum principles are well-known to exclude each other for
helicity $\geq 1$ in standard axiomatic setups of QFT. This
long-standing barrier has been conquered by admitting string-localized
free fields and interactions, and imposing the condition of
string-independence on the $\bS$-matrix.

There are of course fundamental differences between classical
Lagrangians and interaction densities in Quantum Field Theory. The
former are the starting point of a variational principle to derive the
equations of motion. One may freely add terms that are total
divergences {\em without using the equations of motion}. E.g., the
classical expansion of the Einstein Lagrangian as a power series in
the metric deviation field -- see \cite{ScharfLast}. In contrast, the
(perturbative) $\bS$-matrix of QFT is defined as the time-ordered
exponential of the action, i.e., the integral over the {\em
interaction density}
\begin{align}
\label{eq:action}
\int d^4x\, L_{\rm int}(x),
\end{align}
not including a free Lagrangian. This interaction density is a
functional of the free fields, that are autonomously defined on the
Fock space over Wigner's unitary representations of the Poincar\'e
group, automatically operating on a Hilbert space. The equations of
motion of the free fields are consequences of their construction.
(There is a similar split in the path integral approach: the ``free
Lagrangian'' only defines the integration measure, while the integrand
contains only the interaction Lagrangian.) Thus one may freely add
terms to the interaction density that are total derivatives as quantum
operators, i.e., when the free equations of motion are used,
without changing the action \eqref{eq:action}. The last term in
\eqref{eq:L1} is an example. It must not be discarded in classical
field theory.

Now, adding a divergence does affect the $\bS$-matrix, because the
time-ordered exponential of the action is changed. A thorough analysis
can be found in \cite{MundRS23}. There, the point was to add
divergences as to turn a non-renormalizable point-localized
interaction into a renormalizable string-localized one, or a
gauge-theoretic interaction on an indefinite Fock space into a
string-localized interaction in an embedded Hilbert space, as in BRST
theory.

The analysis is not restricted to such cases. The general tenor is
that if two interactions $L_1=K_1+\pa_\mu V_1^\mu$ differ by a
divergence, then the induced interactions will differ by $L_2-K_2$,
such that
\begin{align}
\sfrac12 \bigl(\pa_\mu T V_1^\mu (L'_1+K'_1) - T\pa_\mu V_1^\mu
(L'_1+K'_1)\bigr) +(x\lra x')  = i(L_2-K_2)(x)\delta(x-x').
\end{align}
In other words, the effect of adding a divergence in first order is
compensated by an induced term in second order, and so on, which
itself is not a divergence. It is then established that the resulting
$\bS$-matrices (including the induced higher-order interactions) will
be the same. Thus, physics cannot tell the difference between those
different choices of the interaction density.

To be sure, the story is a bit subtler. Formal expressions like $\int
d^4x\, L(x)$ exist as quantum operators only with a cutoff function
$g(x)$, and the ``adiabatic limit'' $g(x)\uparrow1$ is particularly
laborious in theories with massless fields. It is therefore important
that the equality of $\bS$-matrices hold for arbitrary cutoff
functions -- as shown in \cite{MundRS23} --  before the limit is taken.

\subsection{On the usefulness of string-localized fields}
\label{sec:discussion-on-SLF}

{\em Free} string-localized quantum fields of any spin and helicity
were introduced in \cite{MundSY04,MundSY06}, not least in order to
fill the gap of the (bosonic and fermionic) ``infinite spin'' fields
in the list of quantum fields associated with Wigner's classification
of particles. String-localized {\em charged} fields had been
heuristically introduced much earlier
\cite{Jordan35,Dirac55,Mandelstam62,Steinmann84} in attempts to
formulate Quantum Electrodynamics in terms of gauge-invariant charged
fields. They were indirectly discovered in the axiomatic framework of
algebraic quantum field theory (AQFT) by Buchholz and Fredenhagen
\cite{BF82} in their analysis of charged states.

AQFT does not address charged fields because they are not quantum
observables. On the other hand, it has been shown that the
localization properties of charged states are such as if they were
created from the vacuum by either local or string-localized charged
fields. An actual construction of string-localized Dirac fields going
beyond the perturbative framework was presented in \cite{MundRS21},
mediated by an interaction that involves a string-localized free
Maxwell potential -- the helicity~1 analogue of the present work. The
interacting Dirac field does carry a ``string of electric flux'',
whose direction is described by a profile function $c(e)$ as in the
present work, and whose ``unit weight'' constraint secures the correct
total flux as required by the Gauss Law. Due to the infrared
singularity of photons, different profile functions ``cannot coexist''
-- they do belong to different superselection sectors. Translating
into helicity 2, one expects that matter fields interacting with
quantum gravitons carry a ``string of curvature'', that depends on the
profile function $c(e)$ but ensures that the Komar mass of the
curvature field is that of the matter particle. A spatially constant
profile function (i.e., a completely diffuse string) would give a
quantum state corresponding to the rotationally symmetric
Schwarzschild metric. One also expects a similar infrared
superselection structure, not addressed in this paper. Details of this
picture are under investigation. These issues can be assessed only
because in SQFT charged fields interacting with helicity~1 potentials
(QED) or matter fields interacting with helicity~2 potentials fields
are constructed on a Hilbert space. In contrast, in gauge-theoretic
settings the interacting charged, resp. matter fields are not BRST
invariant, hence they cannot be defined on the BRST Hilbert space.

At any rate, the construction of interacting matter fields along the
lines of \cite{MundRS21} is beyond the scope of the present paper.
Suffice here to mention that \textit{ab initio} string-localized {\em
free} potentials as those in \eqref{eq:def_SL_pot} serve as a tool to
construct ``inherently'' string-localized {\em interacting} fields.

\smallskip

As before stated, the main result of our paper is that the quantum
principles impose on interactions with string-local -- thus positive
-- fields of helicity~2 the same structure as general covariance
imposes on classical Lagrangians. This not only extends a multitude of
analogous cases with helicity (or spin) 1: QED, Yang-Mills, weak
interaction, Higgs models... It is also in line with the fact that
classical field theory can only be a limit in an actual quantum world,%
\footnote{Rather than, reversely, ``to quantize'' as an attempt to get
back the full gamut from its limit.}
and may also be taken as Nature's cue that QG might exist beyond
perturbative QFT on flat background as an autonomous quantum theory.

\appendix

\section{Comparison of SQFT to the BRST formalism}
\label{a:BRST}

A comment on the relation between gauge theory and SQFT is in order.
In the former framework one starts with a canonically quantized
\textit{point-localized} tensor field $h_{\mu\nu}(x)$ -- for instance
in the Feynman gauge, with the propagator given below in
Eq.~\eqref{eq:T0*hh} -- defined on a space with indefinite metric
(``Krein space''). It has spurious degrees of freedom, because
$\eta^{\mu\nu}h_{\mu\nu}=0$ and $\pa^\mu h_{\mu\nu}=0$ cannot hold as
operator identities. The BRST method has the related purposes of
extracting a passage from the Krein space to a Hilbert space and
eliminating spurious degrees of freedom. The BRST operator is a
nilpotent fermionic operator $Q$,%
\footnote{Not to be confused with our $Q_1$ in the $L$-$Q$-pair condition.}
whose kernel is a positive-semidefinite subspace of the Krein space.
For its construction, one has to extend the Krein space by a ghost
Fock space. The BRST variation is the graded commutator:
\begin{align*}
s(A)=i[Q,A]_\pm.
\end{align*}
Then $\mathcal{H}:= \mathrm{Ker}(Q)/\mathrm{Ran}(Q)$ is a Hilbert
space, and only the operators $A\in \mathrm{Ker}(s)$ are well-defined
on $\mathcal{H}$. These are the local observables of the theory.
The BRST variation of the basic field $h_{\mu\nu}$ in this setting
is%
\footnote{In Eq.~(5.1.8) of \cite{ScharfLast} one actually finds the
formula: $s(h^{\mu\nu})=\frac12\big(\pa^\mu u^\nu + \pa^\nu u^\mu -
\eta^{\mu\nu}\pa_\ka u^\ka\big)$. One may pass from this form to ours
by redefining $h_{\mu\nu}$ as $h_{\mu\nu} -\frac12\eta_{\mu\nu}
\bb{h}$ and idetifying $u$ with our~$w$. For the quantum field
$h_{\mu\nu}(c)$ there is no difference because $\bb{h}=0$ and moreover
$(\del w)=0$. Notice that both fields have the same Feynman propagator
\eqref{eq:T0*hh}.}
is a gauge transformation: 
\begin{align}
\label{eq:sh}
s(h_{\mu\nu})=\sfrac12\big(\pa_\mu u_\nu + \pa_\nu u_\mu \big)
\end{align}
where $u_\mu$ is the ghost field. Thus $h_{\mu\nu}$ is not an
observable, but its field strength -- defined as in~\eqref{eq:Fcurlh}
-- is an observable. In fact, on the BRST Hilbert space $\mathcal{H}$
it is isomorphic to our original field $F$ of helicity 2 on the Wigner
Fock space, with only two degrees of freedom according to the
helicities $\vert h\vert =2$.

Let us now discuss how SQFT and BRST are related. This becomes most
transparent in the approaches by Scharf~\cite{ScharfLast} and
Dütsch~\cite{Duetsch05}, who have achieved similar results concerning
self-interactions of helicity 1 and 2, namely the necessity of quartic
interaction terms, by imposing ``perturbative gauge invariance'' in
the BRST setting. The cubic self-interaction density $L_1(x)$ from
gauge theory is not BRST-invariant, but
\begin{align}
\label{eq:sL} 
s\bigl(L_1(x)\bigr) = \pa_\mu T_1^\mu(x),
\end{align}
where $T_1^\mu$ involves the ghost fields. The subsequent analysis is
quite similar to ours (in fact it was a blueprint for ours): the
$\bS$-matrix is BRST-invariant in first order, but not in second order
for the same reason as in~\eqref{S12}. The second-order BRST
obstruction has to be cancelled by a second-order interaction $L_2$.
Which is possible because it has a form analogous to \eqref{eq:SI2}.

One could avoid the BRST obstruction if one had a BRST-invariant
interaction density. Assuming that there is a field $\phi$ on the
Krein space such that $s(\phi_\mu)=-\frac12u_\mu$, then it is feasible
to define the BRST-invariant potential
\begin{align*}
h^\phi_{\mu\nu}(x)= h_{\mu\nu}(x)+
\pa_\mu\phi_\nu+\pa_\nu\phi_\mu.
\end{align*}
In view of \eqref{eq:sh}, $\phi_\mu$ must be some primitive of
$h_{\mu\nu}$. Such a primitive can be constructed with the help of
string integrations over $h_{\mu\nu}$:
\begin{align*}
{}\phi_\mu = I_c^\nu h_{\mu\nu}+
 \sfrac12\del_\mu  I_c^\ka I_c^\la h_{\ka\la}.
\end{align*}
(To verify this claim, one needs Eq.~\eqref{eq:Ic_inverse_del}.)
Writing the interaction $L_1$ in terms of $h^\phi_{\mu\nu}$, one may
discard derivatives -- that are responsible for the right-hand side of
\eqref{eq:sL} -- to get a BRST-invariant interaction. The fields $F$
and $h^\phi$ being BRST-invariant, they ``descend'' to the BRST
Hilbert space $\mathcal{H}$. The restriction of $h^\phi_{\mu\nu}$ to
$\mathcal{H}$ coincides with $h_{\mu\nu}(x;c)$ on the Wigner Fock
space. In particular, $\eta^{\mu\nu}h^\phi_{\mu\nu}$ and $\pa^\mu
h^\phi_{\mu\nu}$ are null fields on the Krein space (their two-point
functions vanish). Finally, the BRST-invariant interaction coincides
with the string-localized $L_1(x;c)$ in~\eqref{eq:L1} up to a
divergence. In this way, one precisely arrives at the setup of the
present paper, where the task is to establish SI rather than BRST
invariance.

The great difference is that we proceeded without the detour through
Krein space and ghosts, from which BRST invariance has to bring us
back home to the physical Hilbert space. A second major advantage of
the string-localized approach is that \textit{interacting field}s (not
considered in the present paper) do exist on our Hilbert space,
whereas in general they are \textit{not} BRST-invariant in the gauge
theoretic setting \cite{MundRS23}. Whether string-localized quantum
field theory is the ``unkown formulation'' for which the BRST fomalism
has been an ``efficient placeholder'' (R. Stora) hangs nevertheless on
unsolved problems of renormalization of loop graphs with internal
stringlike lines.

\section{Graviton propagators and their renormalization}
\label{a:Propagators}

Let us split all propagators into their string-independent part and
the string-dependent ones:
\begin{align}
\label{eq:split}
\vev{T\phi\chi'}= \vev{T_\ast\phi\chi'}+\vev{T_c\phi\chi'}.
\end{align}
Recall the kinematic propagator \eqref{eq:T0hh}. It is composed of the string-integrated
differential operators $E$, $E'$, and $E''$ of \eqref{eq:E} which for
the purposes of \sref{sec:Selfcoupling} and \sref{sec:Cancellation} we rewrite
conveniently as  
\begin{align}
\label{eq:Eab}
E_{\mu\nu} =\eta_{\mu\nu} + a_{(\mu}\pa_{\nu)}, \qquad
E''_{\rho\sg} = \eta_{\rho\sg} + a'_{(\rho}\pa_{\sg)},
\qquad
E'_{\mu\rho} = \eta_{\mu\rho} + (b'_\mu \pa_\rho +b_\rho \pa_\mu),
\end{align}
with (using $\pa'=-\pa$) 
\begin{align}
\label{eq:ab}
\begin{aligned}
&a_\mu=  I_{\mu}+\sfrac12 I^2\pa_\mu, &\qquad &
a'_\rho=-I'_{\rho}+\sfrac12 I'^2\pa_\rho,
\\
&b_\rho=I_{\rho}-\sfrac12 (II')\pa_\rho, &\qquad& 
b'_\mu=  -I'_{\mu}-\sfrac12 (II')\pa_\mu.
\end{aligned}
\end{align}
This yields the decomposition
\begin{align}
\label{eq:T0*hh}
\vev{T_{0,\ast}h_{\mu\nu}h'_{\rho\sg}}&=
\sfrac12\big[\eta_{\mu\rho}\eta_{\nu\sg} +
\eta_{\mu\sg}\eta_{\nu\rho}-\eta_{\mu\nu}\eta_{\rho\sg}\big] D_{0,F},
\\
\label{eq:T0chh}
\vev{T_{0,c}h_{\mu\nu} h'_{\rho\sg}} &=
\sfrac12\big[\pa_{(\mu}a_{\nu), \rho\sg} +
a'_{\mu\nu,(\rho}\pa_{\sg)}\big] D_{0,F}
\end{align}
with
\begin{align}
\label{eq:A}
a_{\nu,\rho\sg} &= \eta_{\nu(\sg}b_{\rho)} + \sfrac12 b'_\nu
b_{(\rho}\pa_{\sg)} + b_\rho b_\sg \pa_\nu - a_\nu(\eta_{\rho\sg}+\sfrac12
a'_{(\rho}\pa_{\sg)}), 
\\
\label{eq:A'}
a'_{\mu\nu,\sg} &= \eta_{\sg(\nu}b'_{\mu)} + \sfrac12 b_\sg
b'_{(\mu}\pa_{\nu)} + b'_\mu b'_\nu \pa_\sg -
a'_\sg(\eta_{\mu\nu}+\sfrac12 a_{(\mu}\pa_{\nu)}).
\end{align}
$\vev{T_{0,\ast}hh}$ is the propagator as one would use in gauge theory.
The salient feature of $\vev{T_{0,r}hh}$, seen in \eqref{eq:T0chh}, is
that each term contains an uncontracted derivative, which plays a
major role in the computations. 

It ensues, with $-(I\pa)=(I'\pa)=1$,
\begin{align}\label{eq:abc}
a_{\nu,\rho\sg}\pa^\nu =\eta_{\rho\sg} + b_{\rho\sg} \square
+ c_{(\rho}\pa_{\sg)}, \qquad
a'_{\mu\nu,\sg}\pa^\sg = \eta_{\mu\nu} + b'_{\mu\nu} \square
+ c'_{(\mu}\pa_{\nu)},
\end{align}
where
\begin{align}
\begin{aligned}
&b_{\rho\sg} = b_\rho b_\sg - \sfrac14  (II')
b_{(\rho}\pa_{\sg)} -\sfrac12 I^2(\eta_{\rho\sg} + \sfrac12 a'_{(\rho}\pa_{\sg)}),
&\quad& c_\rho=\sfrac12 (b_\rho+a'_\rho), \\ 
&b'_{\mu\nu}=b'_\mu b'_\nu - \sfrac14  (II') b'_{(\mu}\pa_{\nu)} -\sfrac12 I^2(\eta_{\mu\nu} +
\sfrac12 a_{(\mu}\pa_{\nu)})
,&\quad& c'_\mu=\sfrac12 (b'_\mu+a_\mu).
\end{aligned}
\label{eq:bc}
\end{align}
These relations become instrumental in the proof of Lemma
\ref{lem:cancel_self}
underlying Theorem
\ref{thm:resol_self_obstr}.

The kinematic propagator violates the trace conditions:
\begin{align}
\label{eq:trace_viol}
\eta^{\mu\nu}\vev{T_0h_{\mu\nu}h'_{\rho\sg}} =
\big[(I_\rho-(II')\pa_\rho)(I_\sg-(II')\pa_\sg)-\sfrac12I^2
E''_{\rho\sg}\big] \square D_{0,F}(x-x'),
 \end{align}
where $\square D_{0,F}=-i\delta(x-x')$. It therefore requires a 
renormalization exploiting the freedom \eqref{eq:freedom_of_reno}, so
as to cancel these traces. Due to the good short distance scaling of 
$h_{\mu\nu}$, this renormalization cannot have a string-independent part.
In order to allow the cancellations in Lemma \ref{lem:cancel_self}
and Lemma \ref{lem:cancel_matter}, it must be of the special form
\begin{align}
\label{eq:p}
\vev{T_{r}h_{\mu\nu} h'_{\rho\sg}} = p_{\mu\nu,\rho\sg} \,
i\delta(x-x') \quad \word{with}
p_{\mu\nu,\rho\sg} =   \pa_{(\mu}p_{\nu),\rho\sg}+p'_{\mu\nu,(\rho}\pa_{\sg)}
\end{align}
exhibiting an uncontracted derivative in each term.

There is a six-parameter family of renormalizations of the form
\eqref{eq:p} that restore the trace conditions:
\begin{align}
\notag
p_{\nu,\rho\sg}=
& \big[\sfrac12 I'_\nu( I_\rho I_\sg  -\sfrac12 I^2\eta_{\rho\sg})
+ \sfrac18 ((II')^2-\sfrac12I^2I'^2)\eta_{\nu(\rho}\pa_{\sg)}
\\ \label{eq:p0}
&\qquad+ \sfrac18\big(I^2 I'_\nu I'_{(\rho}\pa_{\sg)}-2(II') I'_\nu I_{(\rho}\pa_{\sg)}+I'^2 I_\nu I_{(\rho}\pa_{\sg)}\big)\big] \\   \notag       
& +  c_1\big[(I_\nu+I'_\nu)I_{(\rho}I'_{\sg)}                                                                                 +(II')\big(\sfrac12I'_\nu   I_{(\rho}\pa_{\sg)} -\sfrac12 I_\nu I'_{(\rho}\pa_{\sg)}\big)\big]\\ \notag
& +   c_2(II')\big[(I_\nu+I'_\nu)\eta_{\rho\sg}-I_\nu I'_{(\rho}\pa_{\sg)} + I'_\nu I_{(\rho}\pa_{\sg)}\big] \\ \notag       
&+c_3\big[I^2 \eta_{\nu(\rho}I'_{\sg)}-\sfrac12I'^2 I_\nu   I_{(\rho}\pa_{\sg)} -\sfrac12I^2 I'_\nu I'_{(\rho}\pa_{\sg)}  \big]\\ \notag       
&+c_4\big[I'^2 \eta_{\nu(\rho}I_{\sg)}+\sfrac12I'^2 I_\nu I_{(\rho}\pa_{\sg)}+\sfrac12I^2 I'_\nu I'_{(\rho}\pa_{\sg)}\big]\\ \notag       
&+c_5(II')\big[\eta_{\nu(\rho}I_{\sg)}-\sfrac12I_\nu  I'_{(\rho}\pa_{\sg)} - \sfrac12 I'_\nu I_{(\rho}\pa_{\sg)}\big]\\  \label{eq:pc}
&+
c_6(II')\big[ \eta_{\nu(\rho}I'_{\sg)}+\sfrac12I_\nu   I'_{(\rho}\pa_{\sg)} +\sfrac12 I'_\nu I_{(\rho}\pa_{\sg)}\big] ,
\end{align}
and $p'_{\mu\nu,\rho}$ given by the replacements $\sg\to\mu$,
$\nu\lra\rho$, $I'\lra I$ and $\pa\to \pa'=-\pa$.
\begin{lemma}
  \label{lem:cancel_self}
The renormalization part $T_r\underbracket{L_1L}\!'_1$ cancels the
second line \eqref{eq:TL1L1_string} of $T_0\underbracket{L_1L}\!'_1$
for all values of the renormalization constants $c_n$.
\end{lemma}
{\em Proof:} We repeatedly apply the identities (using only
  integrations by part w.r.t.\ $x'$)
  \begin{align} \notag
    I_{(\rho}\pa_{\sg)}i\delta(x-x') U'^{\rho\sg}& \modd 2I_\sg
    i\delta(x-x')\pa'^\sg S' \modd 2I_\sg\pa^\sg
    i\delta(x-x') S'=-2 i\delta(x-x') S',\\ \notag
I'_{(\rho}\pa_{\sg)}i\delta(x-x') U'^{\rho\sg} & \modd 2I'_\sg
    i\delta(x-x')\pa'^\sg S'\modd 2I'_\sg\pa^\sg
    i\delta(x-x')S'= 2i\delta(x-x') S', \\
\pa_\rho\pa_{\sg}i\delta(x-x') U'^{\rho\sg} &\modd i\delta(x-x')
    \pa'_\rho\pa'_\sg U'^{\rho\sg}  =  i\delta(x-x') \square' S'
\modd i\square \delta(x-x') S'.\quad    \label{eq:formulas}
  \end{align}
We first reduce \eqref{eq:TL1L1_string} further by working out
$b_{\rho\sg} i\delta(x-x') U'^{\rho\sg}$ with \eqref{eq:bc} and
\eqref{eq:formulas}: 
  \begin{align*}
 b_{\rho\sg} i\delta(x-x') U'^{\rho\sg}
 &\modd 
(I_\rho I_\sg -\sfrac12   I^2\eta_{\rho\sg})i\delta(x-x')  U'^{\rho\sg} 
\\ &\quad
   + \sfrac12 \big[\big( 3 (II')+I^2\big)+\big((II')^2  -\sfrac12I^2I'^2\big)\square \big]i\delta(x-x')S'. 
  \end{align*}

  Thus, \eqref{eq:TL1L1_string} becomes
  \begin{align} \notag
    \eqref{eq:TL1L1_string} \modd & - S (I_\rho I_\sg -\sfrac12
    I^2\eta_{\rho\sg})i\delta(x-x') U'^{\rho\sg} + (x\lra x')
     \\ \label{eq:TL1L1_string_a}&-S
    \big[(I+I')^2
    +((II')^2-\sfrac12I^2I'^2)\square \big]i\delta(x-x') S'.
  \end{align}
  
On the other hand, we compute the contribution from $\pa_{(\mu}p_{\nu),\rho\sg}$ to
$T_r\underbracket{L_1L}\!'_1$:
\begin{align*}
U^{\mu\nu} \pa_{(\mu}p_{\nu),\rho\sg} \, i\delta(x-x')
  U'^{\rho\sg} \modd 2S  \pa^\nu p_{\nu,\rho\sg}\, i\delta(x-x')
  U'^{\rho\sg}. \end{align*}
Inserting only \eqref{eq:p0} for $p_{\nu,\rho\sg}$, we get
\begin{align*}
S(I_\rho I_\sg -\sfrac12 I^2\eta_{\rho\sg})i\delta(x-x') U'^{\rho\sg}
+ S\cdot \sfrac12\big[ (I+I')^2+
\big((II')^2-\sfrac12I^2I'^2\big)\square \big]i\delta(x-x')
S'.
\end{align*}
Adding the corresponding contribution from
$p'_{\mu\nu,(\rho}\pa_{\sg)}$, we get an exact cancellation of
\eqref{eq:TL1L1_string_a}.

Computing $p^i_{\nu,\rho\sg}\pa^\nu$ for the six structures in \eqref{eq:pc}
and applying then \eqref{eq:formulas}, one very quickly sees that, upon
adding the corresponding term from $p'^i_{\mu\nu,\rho}$, one gets $0$
for all values of the renormalization constants $c_n$. This proves the
Lemma. \hfill $\square$ 

\begin{lemma}\label{lem:cancel_matter}
The renormalization part $T_r\underbracket{L_1L}\!'_{1,\rm mat}$
cancels the 
part \eqref{eq:string} of $T_0\underbracket{L_1L}\!'_{1,\rm mat}$
for arbitrary values of the renormalization constants  $c_1,\dots,c_6$.
\end{lemma}
\begin{proof}
We compute 
\begin{align}\label{inter2}
T_{ r}\underbracket{L_1L}\!'_{1,\rm mat} &\modd \sfrac12 U^{\mu\nu}
 \vev{T_{r}h_{\mu\nu}h'_{\rho\sg}} \Theta'^{\rho\sg} \modd  \sfrac12
U^{\mu\nu} \pa_{(\mu}p_{\nu),\rho\sg} i\delta(x-x') \Theta'^{\rho\sg}.
\end{align}
The operators $p_{\nu,\rho\sg}$ are given in  \eqref{eq:p0}, \eqref{eq:pc}, and the
contributions from $p'_{\mu\nu,(\rho}\pa_{\sg)}$ can be dropped because
$\Theta$ is conserved.  This gives
\begin{align}\notag 
T_{r}\underbracket{L_1L}\!'_{1,\rm mat} &\modd  i\sfrac12
  U^{\mu\nu}  \pa_{(\mu} p_{\nu),\rho\sg}  \delta(x-x')  \Theta'^{\rho\sg} \modd 
 - \pa_\mu U^{\mu\nu}  p_{\nu,\rho\sg}  i\delta(x-x') 
  \Theta'^{\rho\sg} \\
&\modd - i\pa^\nu S  p_{\nu,\rho\sg}  \delta(x-x') 
 \Theta'^{\rho\sg} \modd  iS p_{\nu,\rho\sg}\pa^\nu  \delta(x-x') 
   \Theta'^{\rho\sg} .
\end{align}
We may also drop all derivatives $\pa_\rho$, $\pa_\sg$ 
within $p_{\nu,\rho\sg}$. This leaves 
\begin{align}\label{eq:tildep}\wt p_{\nu,\rho\sg}&=
\sfrac12 I'_\nu I_\rho I_\sg  -\sfrac14 I^2I'_\nu\eta_{\rho\sg} + 
c_1(I_\nu+I'_\nu)I_{(\rho}I'_{\sg)}+  c_2(II')(I_\nu+I'_\nu)\eta_{\rho\sg}
  \\ &+c_3I^2 \eta_{\nu(\rho}I'_{\sg)}+c_4I'^2 \eta_{\nu(\rho}I_{\sg)}+c_5(II')
\eta_{\nu(\rho}I_{\sg)}+c_6(II') \eta_{\nu(\rho}I'_{\sg)}.
\end{align}
Contracting \eqref{eq:tildep} with $\pa^\nu$, we see that the
contributions from $c_1$ and $c_2$ vanish identically and the terms
$c_3,\dots,c_6$ do not contribute after integration by parts due to
the conservation of $\Th'^{\rho\sg}$. The first two terms in
\eqref{eq:tildep} cancel \eqref{eq:string}. This proves the Lemma. 
\end{proof}

\section{Expansion of the classical Einstein action and matter couplings}
\label{a:Expansion}

We may most efficiently choose as starting point an action $\bS_E$ for
gravity introduced by Einstein himself in~\cite{year16} -- consult also
\cite{AlvarezAnero} -- instead of the more popular
Einstein--Hilbert action. The aforesaid action contains, through
the Riemann-Christoffel symbols $\Ga^\bullet_{\bullet\bullet}$, \textit{only first
derivatives} of the metric tensor $g$, which is the fundamental
field. It differs from the latter action by a divergence, and is
therefore classically equivalent. It is of the form:
\begin{align}
\sfrac12\kappa^2\,\bS_E = \int d^4x\,\sqrt{-g(x)}\,g^{\mu\nu}(x)
\bigl(\Ga^\al_{\mu\nu} \Ga^\be_{\al\be} - \Ga^\al_{\mu\be}
\Ga^\be_{\nu\al}\bigr)(x) =: \sfrac12\kappa^2\int d^4x\,S_E(x).
\label{endoftheprinciple}
\end{align}
Here, $-g\equiv -\det g>0$ and $\kappa^{2}=32\pi G$. Since
among other things we undertook to prove that classical general
relativity can be determined from perturbative quantum field theory,
it makes sense to expand the former in terms of~$\kappa$. For our
purposes, we remove the factor $\kappa^2$ in~\eqref{endoftheprinciple} by the
change of variables such that $\Ga^\al_{\mu\nu}$ are $O(\kappa)$. Specifically,
it is convenient to introduce the Goldberg variable
\begin{align}\label{eq:gold}
\go^{\mu\nu} := \sqrt{-g}g^{\mu\nu}; \quad \go_{\mu\nu} :=
\frac{g_{\mu\nu}}{\sqrt{-g}}.
\end{align}
A straightforward computation renders $S_E(x)$ in
\eqref{endoftheprinciple} by means of the
$\go^{\mu\nu}$~\cite[Eq.~(5.5.20)]{ScharfLast}:
\begin{align}
\ka^2S_E(x) &= -\go_{\al\be}\,\del_\nu \go^{\al\mu}\,
\del_\mu\go^{\be\nu} + \sfrac12\go_{\al\rho}\,\go_{\be\sg}\,\go^{\mu\nu}
\del_\mu\go^{\rho\be}\del_\nu\go^{\al\sg}
\notag
\\
&-\sfrac14\go_{\mu\nu}\,\go_{\rho\sg}\,\go^{\al\be}
\del_\al\go^{\mu\nu}\del_\be\go^{\rho\sg}.
\label{eq:indiancircle}
\end{align}

In order to make contact with SQFT theory on Minkowski space, we
parametrize 
\begin{align}
\go^{\mu\nu} = \eta^{\mu\nu} + \ka h^{\mu\nu},
\label{laleyylosprofetas}
\end{align}
where the dynamical field (the ``metric deviation'') $h^{\mu\nu}$ is small only in the sense of
\textit{scattering theory}, that is, it goes to zero at large
distances. That is  precisely what is required in order to
eventually apply the string-independence condition. Therefore we
develop $\ka^2S_E(x)$ following Eq.~\eqref{eq:indiancircle}, with the
help of the inverse metric, of the form: $\go_{\al\be}(x) =
\eta_{\al\be} - \ka h_{\al\be}(x) + \ka^2 h_{\al\rho}(x)h^\rho_\be(x)
- \cdots$.

Collecting in the expression on the right of
\eqref{eq:indiancircle} terms of order $\ka$, one obtains
\cite[Ch.~5.5]{ScharfLast} three-graviton self-couplings with two
derivatives:
\begin{align} \notag
S_E^{(1)}(x) &= \ka h^{\mu\nu} \bigl[\sfrac12\del_\mu h^{\al\be}\del_\nu
h_{\al\be} + \del^\be h_{\mu\al}\del^\al h_{\nu\be} - \del^\rho
h^\al_\mu\del_\rho h_{\nu\al} - \sfrac14 \del_\mu h^\rho_\rho\,
\del_\nu h^\rho_\rho + \sfrac12 \del_\al h^{\mu\nu}\del^\al
h^\rho_\rho\bigr].
\end{align}
Collecting as well terms of order $\ka^2$, one obtains 
four-graviton couplings with two derivatives:
\begin{align}
  \notag
&S_E^{(2)}(x) = -\ka^2\bigl[h_{\al\rho}h^\rho_\be\,\del_\nu
h^{\al\mu}\del_\mu h^{\be\nu} - \sfrac14 h^{\mu\nu} h_{\rho\sg}
\del_\al h_{\mu\nu}\del^\al h^{\rho\sg} -
h^{\mu\nu}h_{\al\rho}\del_\mu h^\rho_\sg\del_\nu h^{\al\sg}
\\ \notag
&+ h_{\rho\be}h^\be_\sg\del_\mu h^{\al\rho}\del^\mu h^\sg_\al + \sfrac12
h_{\al\rho}h_{\be\sg}\del_\mu h^{\al\sg}\del^\mu h^{\be\rho} - \sfrac12
h_{\rho\be}h^\be_\sg\,\del_\al h^{\rho\sg}\del^\al h^\rho_\rho + \sfrac12
h_{\mu\nu}\del_\al h^{\mu\nu}h^{\al\be}\del_\be h^\rho_\rho\bigr].
\end{align} 

Take now the momentuous step of substituting the \textit{SQFT graviton
field} $h^{\mu\nu}(x;c)$ for $h^{\mu\nu}(x)$. This yields the
immediate dividend that the two terms involving the trace
$h^\rho_\rho$ in $S_E^{(1)}$ do vanish. In the understanding that from
now on $h$ denotes the string-localized graviton, we replace
$S_E^{(n)}(x)$ by $\frac{\kappa^n}{n!}L_{n}[h]$ for $n=1,2$ -- recall
the expansion convention \eqref{L}:
\begin{align}
L_1[h] &= h^{\mu\nu}\bigl[\sfrac12\del_\mu h^{\al\be}\del_\nu
h_{\al\be} + \del^\be h_{\mu\al}\del^\al h_{\nu\be} - \del^\rho
h^\al_\mu\del_\rho h_{\nu\al}\bigr](x;c) 
\notag 
\\
&\equiv\bigl[\sfrac12\doublcont{\pa_\mu h \pa_\nu h }h ^{\mu\nu} + (\pa_\nu h h
\pa_\mu h )^{\mu\nu} - \doublcont{h \pa_\mu h \pa^\mu h}\bigr](x;c) ,
\label{eq:eltrucouno}
\\
\sfrac12 L_2[h] &= \bigl[-h_{\al\rho}h^\rho_\be\,\del_\nu h^{\al\mu}\del_\mu
h^{\be\nu} - \sfrac14 h_{\mu\nu} \pa_\al h^{\mu\nu} h_{\rho\sg} \pa^\al
h^\rho\sg - h^{\mu\nu}h_{\al\rho}\del_\mu h^\rho_\sg\del_\nu
h^{\al\sg} 
\notag
\\
&\qquad+ h_{\rho\be}h^\be_\sg\del_\mu h^{\al\rho}\del^\mu
h^\sg_\al+ \sfrac12 h_{\al\rho}h_{\be\sg}\del_\mu h^{\al\sg}\del^\mu
h^{\be\rho}\bigr](x;c)
\notag
\\ 
&\equiv \bigl[-(\pa_\nu h h h \pa_\mu h)^{\mu\nu} -
\sfrac14\bb{h\pa_\mu h}\bb{h\pa^\mu h} - \bb{h\pa_\mu h \pa_\nu
h}h^{\mu\nu} 
\notag
\\
&\qquad+ \bb{h\pa_\mu h \pa^\mu h h}+\sfrac12 \bb{h\pa_\mu h h \pa^\mu
h} \bigr](x;c).
\label{eq:eltrucodos}
\end{align}
Of course these are \eqref{eq:L1} and \eqref{eq:L2_ind} again, modulo
divergences; in particular, the last line of \eqref{eq:eltrucodos} is
separately a divergence. We remark that this is as well similar to the
result of imposing the Hilbert gauge condition:
\begin{align*}
h_\mu^\mu\equiv\eta^{\mu\nu}h_{\mu\nu}=0, \qquad \pa^\mu h_{\mu\nu}=0,
\end{align*}
because $h_{\mu\nu}(x;c)$ satisfies these relations as operator
identities, see \eqref{eq:properties_h}.

Note that because the Riemann--Christoffel symbols
\begin{align*}
\Gamma^\tau_{\ka\la} := g^{\tau\mu}\Gamma_{\mu\ka\la} := \sfrac12
g^{\tau\mu}\big(\pa_\ka g_{\la\mu}+\pa_\la g_{\ka\mu}-\pa_\mu
g_{\ka\la}\big)
\end{align*}
are obviously $O(\kappa)$, also the Riemann tensor is $O(\kappa)$,
specifically
\begin{align}
\label{eq:RF}
R^{\tau}_{\ka\nu\la} := \pa_\nu \Gamma^\tau_{\la\ka}-\pa_\la
\Gamma^\tau_{\nu\ka} + \Gamma^\tau_{\nu\mu}\Gamma^\mu_{\la\ka}-
\Gamma^\tau_{\la\mu}\Gamma^\mu_{\nu\ka} = \sfrac12\kappa
\eta^{\tau\mu}F_{[\mu\ka][\nu\la]} + O(\kappa^2),
\end{align}
where $F$ is defined in terms of the classical field $h$ in the
analogous way as the quantum field strength $F$ is recovered from the
string-localized potential $h(c)$ in \eqref{eq:Fcurlh}. The Ricci tensor
and scalar are $O(\kappa)$ as well -- which is why we work with
the Einstein rather than Einstein-Hilbert Lagrangian. 

For the comparison of the classical matter couplings with the induced
ones from SQFT in Theorem \ref{thm:reproduce_matter}, one has to expand the
generally covariant matter Lagrangians. Because the comparison
requires the substitution of the string-localized quantum field $h_{\mu\nu}(x;c)$ for the classical field
$h_{\mu\nu}(x)$, we shall drop $\bb{h}$ and $\pa_\mu h^{\mu\nu}$
wherever they arise in the classical expansion. Then to $O(\kappa^2)$:
\begin{align}\label{eq:g_exp}
\sqrt{-g}  = 1-\sfrac{\kappa^2}4\bb{hh}, \quad
g^{\mu\nu} = (1+\sfrac{\kappa^2}4\bb{hh})\eta^{\mu\nu} + \kappa
h^{\mu\nu}.
\end{align}
This gives the expansion of the generally covariant scalar Lagrangian:
\begin{align}
\label{eq:sc_exp}L_{g,\phi}
=&~\frac{\sqrt{-g}}2\big(g^{\mu\nu}\pa_\mu\chi\pa_\nu\chi - m^2\chi^2)
\\
=&~L_{0,\phi} + \sfrac\kappa2 \doublcont{h\Theta_{\phi}} +
\sfrac{\kappa^2}2\Bigl[\sfrac12\doublcont{hh}L_{0,\phi} +
\sfrac14\doublcont{hh}\doublcont{\Theta_{\phi}}\Bigr] + \dots \notag
\end{align}
and of the generally covariant Maxwell Lagrangian:
\begin{align}
\label{eq:max_exp}
L_{g,F} = &~\sqrt{-g}\big(-\sfrac14 g^{\mu\nu}
g^{\ka\la}F_{\mu\ka} F_{\nu\la}\big)
\\
= &~L_{0,F} + \sfrac\kappa2\doublcont{h\Theta_{F}} +
\sfrac{\kappa^2}2\big[\sfrac12\doublcont{hh}L_{0,F} - \sfrac12
h^{\mu\nu}h^{\ka\la}F_{\mu\ka}F_{\nu\la} \big] + \dots \notag
\end{align}
The case of the Dirac field is more intricate. The generally covariant Lagrangian is
\begin{align}
  L_{g,\psi}=\sqrt{-g} \cdot \ol\psi\big(\sfrac i2 (\gamma_g^\mu D_\mu -
  \stackrel{\leftarrow} D_\mu\!\!\gamma_g^\mu)-m\big)\psi
\end{align}
where the covariant $\gamma$-matrices 
satisfy $\gamma_g^\mu\gamma_g^\nu+\gamma_g^\nu\gamma_g^\mu=2g^{\mu\nu}$,
and the covariant derivative is defined with the spin
connection. These are defined in terms of a real tetrad
field $e^\mu_\al(x)$ such that
\begin{align}
  g^{\mu\nu}(x)=e^\mu_\al(x) e^\nu_\be(x) \eta^{\al\be}.
\end{align}


The tetrad gives $\gamma_g^\mu$ in terms
of ``flat'' Dirac matrices satisfying
$\gamma^\al\gamma^\be+\gamma^\be\gamma^\al=2\eta^{\al\be}$:
\begin{align}
\label{eq:gamma_g}  \gamma_g^\mu:=e^\mu_\al \gamma^\al
\end{align}
as well as the
coefficients of the spin connection:
\begin{align}
  \omega_{\mu;\al\be}=-\omega_{\mu;\be\al}:= g_{\ka\la}\,e^\ka_\al  D_\mu e^\la_\be=g_{\ka\la}\,e^\ka_\al \big(\pa_\mu e^\la_\be
+ \Gamma^\la_{\mu\nu} e^\nu_\be\big).\end{align}
Then, the covariant derivatives of the spinors are
\begin{align}
 D_\mu\psi = \del_\mu\psi + \sfrac18 \omega_{\mu;\al\be}\,
  [\gamma^\al,\gamma^\be]\psi, \qquad  D_\mu\ol\psi  = \del_\mu\ol\psi
  -\sfrac18 \omega_{\mu;\al\be} \,\ol\psi  [\gamma^\al,\gamma^\be].
\end{align}

Now, the tetrad as given by \cite{OP65} is the ``square root'' of
$g^{\mu\nu}$ given in \eqref{eq:g_exp}. Its expansion is
\begin{align}
  e^\mu_\al = \delta^\mu_\al  + \sfrac\kappa2 h^\mu_\al +
\sfrac{\kappa^2}8\big(\bb{hh}\delta^\mu_\al - (hh)^\mu_\al\big) +
O(\kappa^3),
\end{align}
from which one obtains
\begin{align}\label{eq:spinconn}
  \omega_{\mu;\al\be}= \sfrac\kappa{2}\cdot \pa_{[\al}h_{\be]\mu} + O(\kappa^2).
\end{align}
The first-order term of $L_{g,\psi}$ has contributions from the kinetic terms involving
the derivatives of the Dirac field, and from
$\omega_{\mu;\al\be}$. The latter contribution vanishes
by \eqref{eq:spinconn} together with the
properties of the $\gamma$-matrices and the symmetry of
$h_{\beta\mu}$. The former contribution easily gives $L_{1,\psi} = \frac12\bb{h\Theta_\psi}$. $L_{2,\psi}$
has contributions from $\sqrt{-g}$, from the kinetic terms, and from
$\omega_{\mu;\al\be}$. The first of these is 
$-\frac14\bb{hh}\cdot L_{0,\psi}$. The kinetic
contribution is easily obtained as
$\frac{\kappa^2}8\big(\bb{hh}\bb{\Theta_\psi} -
  \bb{hh\Theta_\psi}\big)$.
Only
the last contribution, involving (by \eqref{eq:gamma_g}) the expansion of
$e^\mu_\nu\omega_{\mu;\al\be}$ to second order,
requires a lengthy computation, with the result
\begin{align*}  \sfrac i8\big[e^\mu_\nu\omega_{\mu;\al\be}\big]^{(2)}\cdot
  \ol\psi(\gamma^\nu\gamma^\al\gamma^\be+\gamma^\al\gamma^\be\gamma^\nu)\psi\modd
  \sfrac i{32}\cdot (h\pa_\nu
h)_{\be\al}\cdot \ol\psi
  (\gamma^\al\gamma^\be\gamma^\nu-\gamma^\nu\gamma^\be\gamma^\al)\psi .\end{align*}
Thus, to second order and modulo derivatives,
\begin{align}
  \label{eq:Dirac_exp}L_{g,\psi} \modd\, & L_{0,\psi} + \sfrac \kappa2
  \bb{h\Theta_\psi}+ \sfrac{\kappa^2}2\big[-\sfrac12\bb{hh}L_{0,\psi}
 \\ \notag &+\sfrac14\bb{hh}\bb{\Theta_\psi} -\sfrac14
  \bb{hh\Theta_\psi} + \sfrac i{16} (h\pa_\nu
h)_{\be\al}\ol\psi (\gamma^\al\gamma^\be\gamma^\nu-\gamma^\nu\gamma^\be\gamma^\al)\psi\big)\big].\end{align}
 Notice that $L_{0,\psi}=0$ when the free quantum Dirac field is
substituted for $\psi$, and the second
  line becomes $-\sfrac14 L_{2,\psi}^1+ \sfrac14 L_{2,\psi}^3 -
  \sfrac18 L_{2,\psi}^4$ in terms of $L_{2,\psi}^n$ as in \eqref{eq:L2n_fermi}.

In all three cases one reads off $L_{1,\rm mat}= \frac12
\doublcont{h\Theta_{\rm mat}}$, with $\Theta_{\rm mat}$ given as in
\eqref{eq:Theta_scalar}--\eqref{eq:Theta_fermi}, and $L_{2,\rm
mat}=[\dots]$ being the expression in the bracket. These expressions
coincide with \eqref{eq:L2_matter_ind} with $L_{2,r,\rm mat}=0$, when
$h_{\mu\nu}(x;c)$ is substituted for the classical field
$h_{\mu\nu}(x)$ and the free quantum matter field for the classical
matter field. This completes the proof of Theorem \ref{thm:reproduce_matter}.

\subsection*{Acknowledgements} C.G. was supported by the National
Science Center of Poland under the grant UMO-2019/35/B/ST1/01651.
J.M.G.B. thanks Enrique Alvarez for suggesting the general problem
treated here, José L. Fernández Barbón and Gernot Akemann for
hospitality and support respectively at IFT-Madrid and Uni-Bielefeld,
where embryonic drafts of this work were broached, and J. C. V\'arilly
for illuminating discussions. K.-H.R. thanks Brian Pitts for interesting
comments and for referring us to ref.\ \cite{OP65}. We all thank Bert
Schroer for comments.

\end{document}